\documentclass[journal,onecolumn,draftcls]{IEEEtran}
\newif\iftwocols
\usepackage{amsmath,amsthm}
\usepackage{tikz}%,tkz-graph}
\usepackage{pgfplots}
\usetikzlibrary{matrix,positioning,fit,calc,shapes}
\pgfplotsset{compat=newest}
\usepackage{graphicx}
\usepackage{color}
\usepackage{epsfig, amsmath, amsthm, amssymb,latexsym,amsmath,amsbsy, rotating}
\usepackage{amsfonts}

\newtheorem{theorem}{Theorem}
\newtheorem{lemma}[theorem]{Lemma}

\newtheorem{proposition}[theorem]{Proposition}

\newtheorem{remark}{Remark}

\theoremstyle{definition}

\let\oldbrace\{
\def\{{\oldbrace\kern0.5pt}

\def\diag{\mathop{\rm diag}\nolimits}%
\newcommand{\ve}{\bf}
\newtheorem{definition}{Definition}
\DeclareFontFamily{U}{futm}{}
\DeclareFontShape{U}{futm}{m}{n}{
  <-> s * [.92] fourier-bb
  }{}
\DeclareSymbolFont{Ufutm}{U}{futm}{m}{n}
\DeclareSymbolFontAlphabet{\mathbb}{Ufutm}
\begin{document}

\title{On Wyner's Common Information in the Gaussian Case}

\author{Erixhen~Sula,~\IEEEmembership{Student Member,~IEEE}~and~Michael~Gastpar,~\IEEEmembership{Fellow,~IEEE}% <-this % stops a space
\thanks{The work in this manuscript was supported in part by the Swiss National Science Foundation under Grant 169294, and by EPFL. The work in this manuscript was partially presented at the {\it 2019 IEEE Information Theory Workshop,} Visby, Sweden and at the {\it 2020 Annual Conference on Information Sciences and Systems,} Princeton, NJ, USA.}
\thanks{E. Sula and M. Gastpar are with the School of Computer and Communication Sciences, {\'E}cole Polytechnique F{\'e}d{\'e}rale de Lausanne (EPFL), Lausanne, Switzerland (email: \{erixhen.sula,michael.gastpar\}@epfl.ch).}% <-this % stops a space
}

\maketitle

\begin{abstract}
Wyner's Common Information and a natural relaxation are studied in the special case of Gaussian random variables.
The relaxation replaces conditional independence by a bound on the conditional mutual information.
The main contribution is the proof that Gaussian auxiliaries are optimal, leading to a closed-form formula.
As a corollary, the proof technique also establishes the optimality of Gaussian auxiliaries for the Gaussian Gray-Wyner network, a long-standing open problem.
\end{abstract}


\IEEEpeerreviewmaketitle

\section{Introduction}

Wyner's Common Information~\cite{Wyner} is a measure of dependence between two random variables.
Its operational significance lies in network information theory problems (including a canonical information-theoretic model of the problem of coded caching) as well as in distributed simulation of shared randomness.
Specifically, for a pair of random variables, Wyner's common information can be described by the search for the most compact third variable that makes the pair conditionally independent. Compactness is measured in terms of the mutual information between the pair and the third variable. The value of Wyner's common information is the minimum of this mutual information.
The main difficulty of Wyner's common information is finding the optimal choice for the third variable.
Indeed, explicit solutions are known only for a handful of special cases, including the binary symmetric double source and the case of jointly Gaussian random variables.

In the same paper~\cite[Section 4.2]{Wyner}, Wyner also proposes a natural relaxation of his common information, obtained by replacing conditional independence with an upper bound on the conditional mutual information. This relaxation is again directly related to network information theory problems, including the Gray-Wyner source coding network~\cite{Gray--Wyner}.
In the present paper, we study this relaxation in the special case of jointly Gaussian random variables.

\subsection{Related Work and Contribution}

The development of Wyner's common information started with the consideration of a particular network source coding problem, now referred to as the Gray-Wyner network~\cite{Gray--Wyner}. From this consideration, Wyner extracted the compact form of the common information in~\cite{Wyner}, initially restricting attention to the case of discrete random variables.
Extensions to continuous random variables are considered in~\cite{Xu--Liu--Chen,Xu--Liu--Chen--2016}, with a closed-form solution for the Gaussian case.
Our work provides an alternative and fundamentally different proof of this same formula (along with a generalization).
In the same line of work Wyner's common information is computed in additive Gaussian channels \cite{Yang--Chen}. A local characterization of Wyner's common information is provided in \cite{Zheng--Xu--2020}, by optimizing over weakly dependent random variables. In~\cite{Witsenhausen} Witsenhausen managed to give closed-form formulas for a class of distributions he refers to as ``L-shaped.'' The concept of Wyner's common information has also been extended using other information measures~\cite{YuT:18}. Other related works include \cite{Veld--Gastpar, Lapidoth--Wigger}. Wyner's common information has many applications, including to communication networks~\cite{Wyner}, to caching~\cite[Section III.C]{Wang--Lim--Gastpar} and to source coding \cite{Satpathy--Cuff}.

%In the Gray-Wyner source coding network~\cite{Gray--Wyner}, there is a single encoder whose task is to describe two sources via three descriptions for two decoders. One description is called the common description, the other two are private descriptions.
%There are two decoders. The first decoder receives the common and one of the private descriptions and is tasked to recover the first source. The second decoder also receives the common description along with the other private description and recovers the second source. The optimal rate-distortion region is known 

For Gaussian sources the Gray-Wyner network~\cite{Gray--Wyner} problem still remains unsolved. A closed form solution is given in \cite{Gray--Wyner} by assuming that the auxiliaries are Gaussian. Partial progress was made in \cite{Akyol--Rose--2014,Xu--Liu--Chen--2016}, when the sum of the common rate and the private rates is exactly equal to the joint rate distortion function.
For this corner case, it is known that Wyner's common information is the smallest rate needed on the common link. In the present paper, we solve the Gray-Wyner network~\cite{Gray--Wyner} for Gaussian sources, encompassing all previous partial results. %also solving the case when the sum of the common rate and the private rates exceeds the joint rate distortion. 
%Yet, a common observation is that the Wyner's common information is the smallest rate that you can send through the common link.

Other variants of Wyner's common information include \cite{Yu--Li--Chen,Kumar--Li--Gamal}. In \cite{Yu--Li--Chen}, the conditional independence constraint is replaced by the conditional maximal correlation constraint, whereas in \cite{Kumar--Li--Gamal}, the mutual information objective is replaced by the entropy. The relaxation of Wyner's common information studied in this paper is different from the above variants in the sense that it can be expressed using only mutual information, and thus it can be expressed as a trade-off curve in the Gray-Wyner region.

The main difficulty in dealing with Wyner's common information is the fact that it is not a convex optimization problem.
Specifically, while the objective is convex, the constraint set is not a convex set : taking convex combinations does not respect the constraint of conditional independence.
The main contributions of our work concern explicit solutions to this non-convex optimization problem in the special case when the underlying random variables are jointly Gaussian. Our contributions include the following:
%\begin{enumerate}
%\item The derivation of the key properties of relaxed version of Wyner's Common Information, most specifically, a Chain Rule for independent pairs (Theorem~\ref{thm:gensplit}).
%\item The closed-form solution for the case of jointly Gaussian random variables . This is accomplished via a novel technique which was originally introduced in \cite{Geng--Nair}.
%%known as factorization of convex envelope, which was originally introduced in \cite{Geng--Nair}.
%\item The closed-form solution for the case of jointly Gaussian random vectors, along with a ``water-filling''-type solution.
%\item The full solution to the long standing problem of Gaussian (lossy) Gray-Wyner network, where we prove that Gaussian auxiliaries are optimal, thus leading to closed-form solutions for the resulting rate-distortion regions.
%\end{enumerate}
\begin{enumerate}
\item We establish an alternative and fundamentally different proof of the well-known formula for (standard) Wyner's common information in the Gaussian case, both for scalars and for vectors. Our proof leverages the technique of factorization of convex envelopes \cite{Geng--Nair}.
\item In doing so, we establish a more general formula for the Gaussian case of a natural relaxation of Wyner's common information. This relaxation was proposed by Wyner. In it, the constraint of conditional independence is replaced by an upper bound on the conditional mutual information. The quantity is of independent interest, for example establishing a rigorous connection between Canonical Correlation Analysis and Wyner's Common Information \cite{Gastpar--Sula2020}.
\item As a corollary, our proof technique also solves a long-standing open problem concerning the Gaussian Gray-Wyner network. Specifically, we establish the optimality of Gaussian auxiliaries for the latter.
\end{enumerate}
 
\subsection{Notation}

We use the following notation. Random variables are denoted by uppercase letters such as $X$ and their realizations by lowercase letters such as $x.$ The alphabets in which they take their values will be denoted by calligraphic letters such as ${\cal X}.$ Random column vectors are denoted by boldface uppercase letters and their realizations by boldface lowercase letters. Depending on the context we will denote the random column vector also as $X^n:=(X_1,X_2,\dots,X_n)$. We denote matrices with uppercase letters, e.g., $A,B,C$. % The $(i,j)$ element of matrix $A$ is denoted by $A_{ij}$ or $[A]_{ij}$ depending on the context. 
For the cross-covariance matrix of $\ve X$ and $\ve Y$, we use the shorthand notation $K_{\ve X \ve Y}$, and for the covariance matrix of a random vector $\ve X$ we use the shorthand notation $K_{\ve X}:= K_{\ve X \ve X}$. In slight abuse of notation, we will let $K_{(X,Y)}$ denote the covariance matrix of the stacked vector $(X,Y)^T.$ We denote the identity matrix of dimension $2\times2$ with $I_2$ and the Kullback-Leibler divergence with $D(.||.)$. The diagonal matrix is denoted by $\diag(.)$. We denote $\log^{+}{(x)}=\max(\log{x},0)$.%$X_{\theta_1}=\frac{X_1+X_2}{\sqrt{2}}$ and $X_{\theta_2}=\frac{X_1-X_2}{\sqrt{2}}$.

\section{Preliminaries}

\subsection{Wyner's Common Information}\label{Sec-WCI-def}

Wyner's common information is defined for two random variables $X$ and $Y$ of arbitrary fixed joint distribution $p(x,y).$

\begin{definition}\label{def-Wyner}
For random variables $X$ and $Y$ with joint distribution $p(x,y),$ Wyner's common information is defined as
\begin{align}
C(X; Y) &=  \inf_{p(w|x,y)} I(X,Y ; W) \mbox{ such that } I(X;Y|W) =0. \label{eq-def-WCI}
\end{align}
\end{definition}

Wyner's common information satisfies a number of interesting properties. We state some of them below in Lemmas~\ref{Lemma-relWyner-basicprops} and~\ref{thm:gensplit} for a generalized definition given in Definition~\ref{def-Wyner-relaxed}.

%stated in the following lemma.
%\begin{lemma}\label{Lemma-Wyner-basicprops}
%Wyner's common information satisfies the following basic properties:
%\begin{enumerate}
%\item $C(X;Y) \ge 0$ with equality if and only if $X$ and $Y$ are independent.
%\item $C(X;Y)\ge I(X;Y).$
%\item Data processing inequality: If $X-Y-Z$ form a Markov chain, then $C(X;Z)\le \min\{ C(X;Y), C(Y;Z)\}.$
%\end{enumerate}
%\end{lemma}
%Proofs are given in Appendix~\ref{app:Lemma-Wyner-basicprops}.

We note that explicit formulas for Wyner's common information are known only for a small number of special cases.
The case of the doubly symmetric binary source is solved completely in~\cite{Wyner} and can be written as
\begin{align}
 C(X;Y) &= 1 + h_b(a_0) - 2 h_b\left(\frac{1-\sqrt{1-2a_0}}{2}\right),    \label{Eq-WCI-binary-WynerParameterization}
\end{align}
where $a_0$ denotes the probability that the two sources are unequal (assuming without loss of generality $a_0\le \frac{1}{2}$).
In this case, the optimizing $W$ is Equation~\eqref{eq-def-WCI} can be chosen to be binary.
Further special cases of discrete-alphabet sources appear in~\cite{Witsenhausen:76}. 

Moreover, when $X$ and $Y$ are jointly Gaussian with correlation coefficient $\rho,$ then
$C(X;Y) = \frac{1}{2} \log \frac{1 + |\rho|}{1-|\rho|}.$
Note that for this example, $I(X;Y) = \frac{1}{2} \log \frac{1}{1-\rho^2}.$
This case was solved in~\cite{Xu--Liu--Chen,Xu--Liu--Chen--2016} using a parameterization of conditionally independent distributions.
We note that an alternative proof follows from our arguments below.

\subsection{A Natural Relaxation of Wyner's Common Information}

Wyner, in~\cite[Section 4.2]{1055346}, defines an auxiliary quantity $\Gamma(\delta_1,\delta_2).$
Starting from this definition, it is natural to introduce the following quantity:

\begin{definition}\label{def-Wyner-relaxed}
For jointly continuous random variables $X$ and $Y$ with joint distribution $p(x,y),$ we define
\begin{align}
C_{\gamma} (X; Y) &=  \inf_{p(w|x,y)} I(X,Y ; W) \mbox{ such that } I(X;Y|W) \le \gamma. \label{Eq-def-Wyner-relaxed}
\end{align}
\end{definition}

With respect to~\cite[Section 4.2]{1055346}, we have that $C_{\gamma} (X; Y) = H(X,Y)-\Gamma(0,\gamma).$
Comparing Definitions~\ref{def-Wyner} and~\ref{def-Wyner-relaxed},
we see that in $C_{\gamma} (X; Y),$ the constraint of conditional independence is relaxed into an upper bound on the conditional mutual information.
Specifically, for $\gamma=0,$ we have $C_0(X;Y) = C(X;Y),$ the regular Wyner's common information.
In this sense, it is tempting to refer to $C_{\gamma} (X; Y)$ as {\it relaxed}  Wyner's common information.
The following lemma summarizes some basic properties.

\begin{lemma}\label{Lemma-relWyner-basicprops}
$C_{\gamma} (X; Y)$ satisfies the following basic properties:
\begin{enumerate}
%\item the cardinality of ${\cal W}$ may be restricted to $|{\cal W}| \le |{\cal X}| |{\cal Y}| + 1.$
%\item $C_\gamma(X;Y) \ge 0$ with equality if and only if $\gamma \ge I(X;Y).$
\item $C_{\gamma} (X; Y) \ge \max\{ I(X;Y)-\gamma, 0\}.$
\item Data processing inequality: If $X-Y-Z$ form a Markov chain, then $C_\gamma(X;Z)\le \min\{ C_\gamma(X;Y), C_\gamma(Y;Z)\}.$
\item $C_\gamma(X;Y)$ is a convex and continuous function of $\gamma$ for $\gamma \ge 0.$
\item If $Z$ is independent of $(X,Y),$ then $C_\gamma((X,Z);Y)= C_{\gamma} (X;Y).$
%\item If $f(\cdot)$ and $g(\cdot)$ are one-to-one functions, then $C_{\gamma}(f(X); g(Y)) = C_\gamma(X;Y).$
%\item For discrete $X,$ we have $C_\gamma(X;X) = \max\{H(X)-\gamma,0\}.$
\end{enumerate}
\end{lemma}

Proofs are given in Appendix~\ref{app:Lemma-relWyner-basicprops}.

A further property of $C_{\gamma} (X; Y)$ is a tensorization result for independent pairs, which we will use below to solve the case of the Gaussian vector source.

\begin{lemma}[Tensorization] \label{thm:gensplit}
Let $\{ (X_i,Y_i) \}_{i=1}^n$ be $n$ independent pairs of random variables. Then
\begin{align}
C_{\gamma} (X^n; Y^n) &=  \min_{\{\gamma_i\}_{i=1}^n : \sum_{i=1}^n \gamma_i=\gamma}  \sum_{i=1}^n C_{\gamma_i}(X_i; Y_i).  \label{Eqn-thm:gensplit}
\end{align}
\end{lemma}
%\begin{proof}
The proof is given in Appendix \ref{app:chainsplit}.
%\end{proof}
The lemma has an intuitive interpretation in $\mathbb{R}^2$ plane. If we express $C_{\gamma} (X^n; Y^n)$ as a region in $\mathbb{R}^2$, which is determined by $(\gamma,C_{\gamma} (X^n; Y^n))$, then the computation of $(\gamma,C_{\gamma} (X^n; Y^n))$ is simply the Minkowski sum of the individual regions which are determined by $(\gamma_i,C_{\gamma_i} (X_i; Y_i))$.

\begin{remark}
Not surprisingly, for probabilistic models beyond independent pairs of random variables,
one cannot generally order the quantities on the left and right hand sides in Equation~\eqref{Eqn-thm:gensplit},  respectively.
To see that the right hand side in Equation~\eqref{Eqn-thm:gensplit} can be an upper bound to the left hand side, suppose first that $X_1=X_2$ and $Y_1=Y_2.$
Then, $C_\gamma(X_1,X_2;Y_1,Y_2) =C_\gamma(X_1;Y_1)\le C_{\gamma/2}(X_1;Y_1)+C_{\gamma/2}(X_2;Y_2)=\min_{\gamma_1+\gamma_2\le\gamma} C_{\gamma_1}(X_1; Y_1) + C_{\gamma_2}(X_2; Y_2),$
since $C_{\gamma}(X;Y)$ is a non-increasing function of $\gamma.$
By contrast, to see that the right hand side in Equation~\eqref{Eqn-thm:gensplit} can be a lower bound to the left hand side, consider now binary random variables and let $X_1$ and $Y_1$ be independent and uniform.
Let $X_2=X_1\oplus Z$ and $Y_2=Y_1\oplus Z,$ where $Z$ is binary uniform and independent, and $\oplus$ denotes modulo-addition.
Then, $C_\gamma(X_1;Y_1)=C_\gamma(X_2;Y_2)=0,$ while $C_\gamma(X_1,X_2;Y_1,Y_2)\ge C_\gamma(Z;Z) = 1-\gamma,$
where the inequality is due to the Data Processing Inequality, i.e., Item 2) of Lemma~\ref{Lemma-relWyner-basicprops},
and it is straightforward to establish that for discrete random variables $Z,$ we have $C_\gamma(Z;Z) = H(Z)-\gamma.$
%$I(W;X_1,X_2,Y_1,Y_2)=I(W;X_1,Y_1,Z,Z) \geq I(W;Z,Z)$ and $I(X_1,X_2;Y_1,Y_2|W)=I(X_1,Z;Y_1,Z|W) \geq I(Z;Z|W).$
%to the Data Processing Inequality, i.e., Item 4) of Lemma~\ref{Lemma-relWyner-basicprops}.
\end{remark}

\section{The Scalar Gaussian Case}

One of the main technical contributions of this work is a closed-form formula for $C_{\gamma}(X;Y)$ in the case where $X$ and $Y$ are jointly Gaussian.

\begin{theorem} \label{thm:scalar}
When $X$ and $Y$ are jointly Gaussian with correlation coefficient $\rho,$ then
\begin{align}
C_{\gamma}(X;Y) &= \frac{1}{2} \log^+ \left( \frac{1 + |\rho|}{1-|\rho|} \cdot \frac{1 - \sqrt{1-e^{-2\gamma}}}{1 + \sqrt{1-e^{-2\gamma}}} \right).
\end{align}
\end{theorem}
The proof is given below in Section~\ref{sec:scalar}.

%\begin{remark}[Uniqueness]\label{remark-uniqueness}
%We note that in the Gaussian case, $C_{\gamma}(X;Y)$ is attained by selecting $W$ in Equation~\eqref{Eq-def-Wyner-relaxed} jointly Gaussian with $X$ and $Y.$
%Note, however, that this is not a unique choice. In particular, any optimizing $W$ may be replaced by $f(W),$ for any one-to-one function $f(\cdot),$ since in that case,
%both $I(X,Y;W) = I(X,Y;f(W))$ and $I(X;Y|W)=I(X;Y|f(W)).$
%\end{remark}

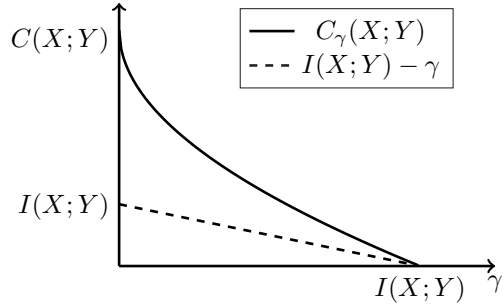
\begin{figure}[ht]% 
\centering
\scalebox{1}{% This file was created by matlab2tikz.
\begin{tikzpicture}
\draw[->, line width=1pt](0,0)--(5,0);
\draw[->, line width=1pt](0,0)--(0,3.5);
\draw[black,thick,anchor=north] (5,0) node{$\gamma$};
\draw[black,thick,anchor=north] (4,0) node{$I(X;Y)$};
\draw[black,thick,anchor=east] (0,3) node{$C(X;Y)$};
\draw[black,thick,anchor=east] (0,0.8) node{$I(X;Y)$};
\pgfplotsset{ticks=none}
\begin{axis}[%
hide y axis,
hide x axis,
width=5in,
height=4in,
xmin=0,
xmax=0.4,
ymin=0,
ymax=1.5,
xlabel={$\gamma$},
legend style={at={(0.4,0.4)}, legend cell align=center, align=center, draw=white!15!black}
]
\addplot [line width=1pt, color=black, mark=--, mark options={solid, black}]
  table[row sep=crcr]{%
0	0.549306144334055\\
0.001	0.504577330851589\\
0.002	0.48603950717294\\
0.003	0.471807741773195\\
0.004	0.459803784846008\\
0.005	0.449222790204602\\
0.006	0.439652055528613\\
0.007	0.430846458616641\\
0.008	0.422646315849692\\
0.009	0.414940729296579\\
0.01	0.40764896840636\\
0.011	0.400710098342209\\
0.012	0.394076786712141\\
0.013	0.387711398833854\\
0.014	0.381583425740229\\
0.015	0.375667727857034\\
0.016	0.369943298859309\\
0.017	0.364392372959695\\
0.018	0.358999765752831\\
0.019	0.353752378014277\\
0.02	0.348638815768086\\
0.021	0.34364909496579\\
0.022	0.338774408828048\\
0.023	0.334006942327358\\
0.024	0.329339722638976\\
0.025	0.324766497387931\\
0.026	0.320281634627587\\
0.027	0.315880039989489\\
0.028	0.311557087533705\\
0.029	0.307308561628681\\
0.03	0.30313060778399\\
0.031	0.299019690806169\\
0.032	0.294972558987265\\
0.033	0.290986213296099\\
0.034	0.287057880743821\\
0.035	0.283184991252743\\
0.036	0.279365157481271\\
0.037	0.275596157156003\\
0.038	0.271875917540432\\
0.039	0.268202501732758\\
0.04	0.264574096536266\\
0.041	0.260989001687187\\
0.042	0.257445620258891\\
0.043	0.253942450089151\\
0.044	0.250478076100256\\
0.045	0.247051163400921\\
0.046	0.24366045107487\\
0.047	0.240304746574386\\
0.048	0.236982920648341\\
0.049	0.233693902743737\\
0.05	0.230436676827826\\
0.051	0.227210277584757\\
0.052	0.22401378694651\\
0.053	0.220846330922901\\
0.054	0.217707076699762\\
0.055	0.214595229978062\\
0.056	0.21151003252998\\
0.057	0.208450759950697\\
0.058	0.205416719587092\\
0.059	0.202407248626617\\
0.06	0.199421712331485\\
0.061	0.196459502404876\\
0.062	0.193520035477317\\
0.063	0.190602751702596\\
0.064	0.187707113453683\\
0.065	0.184832604110089\\
0.066	0.181978726928954\\
0.067	0.179145003992894\\
0.068	0.176330975228345\\
0.069	0.173536197488699\\
0.07	0.17076024369709\\
0.071	0.16800270204416\\
0.072	0.165263175236547\\
0.073	0.162541279792245\\
0.074	0.159836645379292\\
0.075	0.157148914194606\\
0.076	0.154477740379991\\
0.077	0.151822789472672\\
0.078	0.149183737887863\\
0.079	0.146560272431142\\
0.08	0.143952089838538\\
0.081	0.141358896342451\\
0.082	0.138780407261643\\
0.083	0.136216346613686\\
0.084	0.133666446748395\\
0.085	0.131130448000857\\
0.086	0.128608098362815\\
0.087	0.126099153171202\\
0.088	0.123603374812778\\
0.089	0.121120532443843\\
0.09	0.118650401724098\\
0.091	0.116192764563804\\
0.092	0.113747408883422\\
0.093	0.111314128384998\\
0.094	0.108892722334592\\
0.095	0.106482995355118\\
0.096	0.104084757228979\\
0.097	0.101697822709951\\
0.098	0.0993220113437749\\
0.099	0.0969571472969869\\
0.1	0.0946030591935194\\
0.101	0.0922595799586462\\
0.102	0.0899265466698777\\
0.103	0.0876038004144262\\
0.104	0.0852911861528944\\
0.105	0.0829885525888578\\
0.106	0.0806957520440319\\
0.107	0.078412640338736\\
0.108	0.0761390766773796\\
0.109	0.0738749235387187\\
0.11	0.0716200465706396\\
0.111	0.0693743144892443\\
0.112	0.0671375989820256\\
0.113	0.0649097746149301\\
0.114	0.0626907187431213\\
0.115	0.0604803114252628\\
0.116	0.0582784353411551\\
0.117	0.0560849757125667\\
0.118	0.0538998202271082\\
0.119	0.0517228589650097\\
0.12	0.0495539843286661\\
0.121	0.0473930909748236\\
0.122	0.0452400757492883\\
0.123	0.0430948376240425\\
0.124	0.0409572776366614\\
0.125	0.0388272988319277\\
0.126	0.0367048062055493\\
0.127	0.034589706649885\\
0.128	0.0324819089015959\\
0.129	0.0303813234911347\\
0.13	0.0282878626939989\\
0.131	0.0262014404836709\\
0.132	0.0241219724861751\\
0.133	0.0220493759361835\\
0.134	0.019983569634608\\
0.135	0.0179244739076175\\
0.136	0.0158720105670194\\
0.137	0.0138261028719536\\
0.138	0.0117866754918457\\
0.139	0.00975365447056828\\
0.14	0.0077269671917613\\
0.141	0.00570654234527018\\
0.142	0.00369230989465424\\
0.143	0.00168420104572345\\
};
\addlegendentry{$C_{\gamma}(X;Y)$}

\addplot [line width=1pt, dashed]
table[row sep=crcr]{%
0	0.14384103622589\\
0.001	0.14284103622589\\
0.002	0.14184103622589\\
0.003	0.14084103622589\\
0.004	0.13984103622589\\
0.005	0.13884103622589\\
0.006	0.13784103622589\\
0.007	0.13684103622589\\
0.008	0.13584103622589\\
0.009	0.13484103622589\\
0.01	0.13384103622589\\
0.011	0.13284103622589\\
0.012	0.13184103622589\\
0.013	0.13084103622589\\
0.014	0.12984103622589\\
0.015	0.12884103622589\\
0.016	0.12784103622589\\
0.017	0.12684103622589\\
0.018	0.12584103622589\\
0.019	0.12484103622589\\
0.02	0.12384103622589\\
0.021	0.12284103622589\\
0.022	0.12184103622589\\
0.023	0.12084103622589\\
0.024	0.11984103622589\\
0.025	0.11884103622589\\
0.026	0.11784103622589\\
0.027	0.11684103622589\\
0.028	0.11584103622589\\
0.029	0.11484103622589\\
0.03	0.11384103622589\\
0.031	0.11284103622589\\
0.032	0.11184103622589\\
0.033	0.11084103622589\\
0.034	0.10984103622589\\
0.035	0.10884103622589\\
0.036	0.10784103622589\\
0.037	0.10684103622589\\
0.038	0.10584103622589\\
0.039	0.10484103622589\\
0.04	0.10384103622589\\
0.041	0.10284103622589\\
0.042	0.10184103622589\\
0.043	0.10084103622589\\
0.044	0.0998410362258905\\
0.045	0.0988410362258905\\
0.046	0.0978410362258905\\
0.047	0.0968410362258905\\
0.048	0.0958410362258904\\
0.049	0.0948410362258904\\
0.05	0.0938410362258904\\
0.051	0.0928410362258904\\
0.052	0.0918410362258904\\
0.053	0.0908410362258905\\
0.054	0.0898410362258905\\
0.055	0.0888410362258905\\
0.056	0.0878410362258905\\
0.057	0.0868410362258905\\
0.058	0.0858410362258905\\
0.059	0.0848410362258905\\
0.06	0.0838410362258905\\
0.061	0.0828410362258905\\
0.062	0.0818410362258905\\
0.063	0.0808410362258905\\
0.064	0.0798410362258904\\
0.065	0.0788410362258904\\
0.066	0.0778410362258904\\
0.067	0.0768410362258904\\
0.068	0.0758410362258904\\
0.069	0.0748410362258904\\
0.07	0.0738410362258904\\
0.071	0.0728410362258904\\
0.072	0.0718410362258904\\
0.073	0.0708410362258904\\
0.074	0.0698410362258904\\
0.075	0.0688410362258904\\
0.076	0.0678410362258904\\
0.077	0.0668410362258904\\
0.078	0.0658410362258904\\
0.079	0.0648410362258904\\
0.08	0.0638410362258904\\
0.081	0.0628410362258904\\
0.082	0.0618410362258904\\
0.083	0.0608410362258904\\
0.084	0.0598410362258904\\
0.085	0.0588410362258904\\
0.086	0.0578410362258904\\
0.087	0.0568410362258904\\
0.088	0.0558410362258904\\
0.089	0.0548410362258904\\
0.09	0.0538410362258904\\
0.091	0.0528410362258904\\
0.092	0.0518410362258904\\
0.093	0.0508410362258904\\
0.094	0.0498410362258904\\
0.095	0.0488410362258904\\
0.096	0.0478410362258904\\
0.097	0.0468410362258904\\
0.098	0.0458410362258904\\
0.099	0.0448410362258904\\
0.1	0.0438410362258904\\
0.101	0.0428410362258904\\
0.102	0.0418410362258904\\
0.103	0.0408410362258904\\
0.104	0.0398410362258904\\
0.105	0.0388410362258904\\
0.106	0.0378410362258904\\
0.107	0.0368410362258904\\
0.108	0.0358410362258904\\
0.109	0.0348410362258904\\
0.11	0.0338410362258904\\
0.111	0.0328410362258904\\
0.112	0.0318410362258904\\
0.113	0.0308410362258904\\
0.114	0.0298410362258904\\
0.115	0.0288410362258904\\
0.116	0.0278410362258904\\
0.117	0.0268410362258904\\
0.118	0.0258410362258904\\
0.119	0.0248410362258904\\
0.12	0.0238410362258904\\
0.121	0.0228410362258904\\
0.122	0.0218410362258904\\
0.123	0.0208410362258904\\
0.124	0.0198410362258904\\
0.125	0.0188410362258905\\
0.126	0.0178410362258904\\
0.127	0.0168410362258904\\
0.128	0.0158410362258904\\
0.129	0.0148410362258904\\
0.13	0.0138410362258904\\
0.131	0.0128410362258904\\
0.132	0.0118410362258904\\
0.133	0.0108410362258904\\
0.134	0.00984103622589044\\
0.135	0.00884103622589044\\
0.136	0.00784103622589044\\
0.137	0.00684103622589044\\
0.138	0.00584103622589044\\
0.139	0.00484103622589044\\
0.14	0.00384103622589044\\
0.141	0.00284103622589044\\
0.142	0.00184103622589044\\
0.143	0.000841036225890435\\
};
\addlegendentry{$I(X;Y) -\gamma$}
\end{axis}
\end{tikzpicture}%
}
\caption{$C_{\gamma}(X;Y)$ for jointly Gaussian $X$ and $Y$ for the case $\rho=1/2,$ thus, we have $C(X;Y)=\log \sqrt{3}$ and $I(X;Y)=\log (2/\sqrt{3}).$ The dashed line is the lower bound from Lemma~\ref{Lemma-relWyner-basicprops}, Item 1).}
\label{fig:lowerbound}
\end{figure}

%\begin{remark}[Operational Significance]
%While relaxed Wyner's common information as in Equation~\eqref{Eq-def-Wyner-relaxed} is well-defined for the case of Gaussian sources $X$ and $Y,$
%we note that the specifics of the discussion of operational significance presented in Section~\ref{Sec-operational} concerned perfect reconstruction of the sources. This does not apply to the case of Gaussian sources. Instead, Section~\ref{Sec-GaussianGrayWyner} below will discuss a lossy version of our problem, and thus, establish operational significance in the case of Gaussian sources.
%\end{remark}

\subsection{Preliminary Results for the Proof of Theorem~\ref{thm:scalar}} \label{sec:scalar:prelim}

The following results are used as intermediate tools in the proof of the main results.
\begin{theorem} \label{Thm:Hypercontract}
For $K \succeq 0$, $0 < \lambda < 1$, there exists a $0\preceq K^{\prime} \preceq K$ and $(X^{\prime},Y^{\prime})\sim \mathcal{N}(0,K^{\prime})$ such that $(X, Y)\sim p_{X,Y}$ with covariance matrix $K$ the following inequality holds:
\begin{align}
\inf_W h(Y|W)+h(X|W) - (1+\lambda) h(X, Y|W)  \geq h(Y^{\prime})+ h(X^{\prime}) -(1+\lambda)h(X^{\prime}, Y^{\prime}) . \label{Eq-Thm:Hypercontract}
\end{align}
\end{theorem}
\begin{proof}
The theorem is a consequence of \cite[Theorem~2]{Hyper_Gauss}, for a specific choice of $p=\frac{1}{\lambda}+1$. The rest of the proof is given in Appendix \ref{app:HyperGauss}.
\end{proof}

To leverage Theorem~\ref{Thm:Hypercontract}, we need to understand the covariance matrix $K'.$
In~\cite{Hyper_Gauss}, the right hand side in Equation~\eqref{Eq-Thm:Hypercontract} is further lower bounded as $h(Y)+ h(X) -(1+\lambda)h(X, Y)$, where $(X,Y)\sim \mathcal{N}(0,K)$ (correlation coefficient of matrix $K$ is $\rho$ and the diagonal entries are unity), which holds for $\lambda < \rho.$
This choice establishes the hypercontractivity bound $(1+\rho)I(W;X,Y) \geq I(W;X)+ I(W;Y)$ (for jointly Gaussian $X,Y$ and any $W$).
Unfortunately, for the problem of Wyner's common information, this leads to a loose lower bound, which can be seen as follows:
\begin{align}
C_{\gamma=0}(X;Y)&= \inf_{p(w|x,y):I(X;Y|W) =0} I(X,Y;W) \\
&= \inf_{p(w|x,y):I(X,Y;W)+I(X;Y)-I(W;X)-I(W;Y) =0} I(X,Y;W) \\
& \geq \inf_{p(w|x,y):I(X;Y)-\rho I(X,Y;W) \leq 0} I(X,Y;W)  \label{eqn:appmutualcontractivity}\\
&=\inf_{p(w|x,y):I(X,Y;W) \geq \frac{I(X;Y)}{\rho}} I(X,Y;W) \\
&=\frac{I(X;Y)}{\rho}=\frac{1}{2}\frac{\log{\frac{1}{1-\rho^2}}}{\rho},
\end{align}
where (\ref{eqn:appmutualcontractivity}) follows from $(1+\rho)I(W;X,Y) \geq I(W;X)+ I(W;Y)$.

We now show that by a different lower bound on the right hand side in Equation~\eqref{Eq-Thm:Hypercontract},
we can indeed get a tight lower bound for the problem of Wyner's common information as well as its relaxation $C_{\gamma}(X;Y).$
Specifically, we have the following lower bound:

\begin{lemma} \label{lem:lemmasymmery}
For $(X^{\prime},Y^{\prime})\sim \mathcal{N}(0,K^{\prime})$, the following inequality holds
\begin{align}
\min_{K^{\prime}: 0 \preceq K^{\prime} \preceq \begin{pmatrix} 1 & \rho \\ \rho &1  \end{pmatrix}} h(X^{\prime})+h(Y^{\prime})-(1+\lambda)h(X^{\prime},Y^{\prime}) \geq  \frac{1}{2} \log{\frac{1}{1-\lambda^2}}-\frac{\lambda}{2} \log{(2\pi e)^2\frac{(1-\rho)^2(1+\lambda)}{1-\lambda}},
\end{align}
where $\lambda \leq \rho$.
\end{lemma}
\begin{proof}
The proof is given in Appendix \ref{App:lowerboundmainRWCI}.
\end{proof}

%\subsection{Exact Computation of Relaxed Wyner's Common Information for Gaussian source} \label{sec:scalar} 
\subsection{Proof of Theorem~\ref{thm:scalar}} \label{sec:scalar} 
The proof of the converse for Theorem \ref{thm:scalar} involves two main steps. In this section, we prove that one optimal distribution is  jointly Gaussian via a variant of the factorization of convex envelope. Then, we tackle the resulting optimization problem with Lagrange duality.  Let us start form the lower bound first.
\begin{lemma}  \label{lemma:LBRWCI}
When X and Y are jointly Gaussian with correlation coefficient $\rho$ and unit variance, then $C_{\gamma}(X;Y)\geq \frac{1}{2}\log^+{ \left(\frac{1+|\rho|}{1-|\rho|} \cdot \frac{1-\sqrt{1-e^{-2\gamma}}}{1+\sqrt{1-e^{-2\gamma}}} \right) } $.
\end{lemma}
\begin{proof}
The lower bound is derived in the following lines
\begin{align}
C_{\gamma}(X;Y)&=\inf_{W:I(X;Y|W) \leq \gamma} I(X,Y;W) \\
& \geq \inf_{W} (1+\mu)I(X,Y;W)-\mu I(X;W) -\mu I(Y;W) +\mu I(X;Y) -\mu \gamma \label{eqn:alllambda}  \\
&= h(X,Y) -\mu \gamma +  \mu \inf_{W}h(X|W)+h(Y|W) -(1+\frac{1}{\mu})h(X,Y|W) \\
&\geq h(X,Y) -\mu \gamma + \mu \min_{K^{\prime}: 0 \preceq K^{\prime} \preceq \begin{pmatrix} 1 & \rho \\ \rho &1  \end{pmatrix}} h(X^{\prime})+h(Y^{\prime})-(1+\frac{1}{\mu})h(X^{\prime},Y^{\prime}) \label{eqn:thm2sim} \\
& \geq \frac{1}{2} \log{(2\pi e)^2 (1-\rho^2)} -\mu \gamma + \frac{\mu}{2} \log{\frac{\mu^2}{\mu^2-1}}-\frac{1}{2} \log{(2\pi e)^2\frac{(1-\rho)^2(\mu+1)}{\mu-1}}  \label{eqn:lastexp} \\
& \geq \log^+{ \left( \frac{1+|\rho|}{1-|\rho|} \cdot \frac{1-\sqrt{1-e^{-2\gamma}}}{1+\sqrt{1-e^{-2\gamma}}} \right) } \label{eqn:lastexp2}
\end{align}
where (\ref{eqn:alllambda}), is a bound for all $\mu \geq 0$; (\ref{eqn:thm2sim}) follows from Theorem \ref{Thm:Hypercontract} where $(X^{\prime},Y^{\prime}) \sim \mathcal{N}(0,K^{\prime})$, $\mu:=\frac{1}{\lambda}$ and for the assumption $0 < \lambda <1$ to be satisfied we need $\mu > 1$; (\ref{eqn:lastexp}) follows from Lemma \ref{lem:lemmasymmery} for $\mu \geq \frac{1}{\rho}$ and (\ref{eqn:lastexp2}) follows by maximizing the function 
\begin{align}
g(\mu):=\frac{1}{2} \log{(2\pi e)^2 (1-\rho^2)} -\mu \gamma + \frac{\mu}{2} \log{\frac{\mu^2}{\mu^2-1}}-\frac{1}{2} \log{(2\pi e)^2\frac{(1-\rho)^2(\mu+1)}{\mu-1}},
\end{align} 
for $\mu \geq \frac{1}{\rho}$. Now we need to choose the tightest bound where $\mu \geq \frac{1}{\rho}$, which is $\max_{\mu \geq \frac{1}{\rho}}  g(\mu)$ and function $g$ is concave in $\mu$,
\begin{align}
\frac{\partial^2 g}{\partial \mu^2}&=-\frac{1}{\mu(\mu^2-1)} < 0.
\end{align}
By studying the monotonicity we obtain
\begin{align}
\frac{\partial g}{\partial \mu}&=-\frac{1}{2}\log{\frac{\mu^2-1}{\mu^2}}  -\gamma, 
\end{align}
since the function is concave the maximum has to be when the derivative vanishes which leads to the optimal solution $\mu_*=\frac{1}{\sqrt{1-e^{-2\gamma}}}$, where $\mu_* \geq \frac{1}{\rho}$. Substituting for the optimal $\mu_*$ we obtain 
\begin{align}
C_{\gamma}(X;Y) \geq g \left( \frac{1}{\sqrt{1-e^{-2\gamma}}} \right)=\frac{1}{2}\log^+{ \left( \frac{1+\rho}{1-\rho} \cdot  \frac{1-\sqrt{1-e^{-2\gamma}}}{1+\sqrt{1-e^{-2\gamma}}} \right)}.
\end{align}
\end{proof}

Now let us move the attention to the upper bound. Let us assume (without loss of generality) that $X$ and $Y$ have unit variance and are non-negatively correlated
with correlation coefficient $\rho \ge 0.$
Since they are jointly Gaussian, we can express them as
\begin{align}
X&=\sigma W+\sqrt{1-\sigma^2}N_X \\
Y&=\sigma W+\sqrt{1-\sigma^2}N_Y, 
\end{align}
where $W, N_X, N_Y$ are jointly Gaussian, and where $W\sim \mathcal{N}(0,1)$ is independent of $(N_X, N_Y).$
Letting the covariance of the vector $(N_X, N_Y)$ be
\begin{align}
K_{(N_X,N_Y)}=\begin{pmatrix} 
1 & \alpha  \\
\alpha & 1 \\
\end{pmatrix}
\end{align}
for some $0 \le \alpha \le \rho,$
we find that we need to choose $\sigma^2=\frac{\rho-\alpha}{1-\alpha}.$
Specifically, let us select $\alpha=\sqrt{1-e^{-2\gamma}},$ for some $0\le \gamma \le \frac{1}{2}\log\frac{1}{1-\rho^2}.$
For this choice, we find $I(X;Y|W)=\gamma$ and 
\begin{align}
I(X,Y;W)= \frac{1}{2} \log \frac{(1+\rho)(1-\alpha)}{(1-\rho)(1+\alpha)}.
\end{align}

%%%%%%%%%%%%%%%%%%%%%%%%%%%%%%%%%%%%%%%%%%%%%%%%%%%%%%%%%%%%%%%%%%%%%%%%
%%%%%%%%%%%%%%%%%%%%%%%%%%%%%%%%%%%%%%%%%%%%%%%%%%%%%%%%%%%%%%%%%%%%%%%%
\section{The Vector Gaussian Case}
%%%%%%%%%%%%%%%%%%%%%%%%%%%%%%%%%%%%%%%%%%%%%%%%%%%%%%%%%%%%%%%%%%%%%%%%

In this section, we consider the case where $\ve X$ and $\ve Y$ are jointly Gaussian random vectors.
The key observation is that in this case, there exist invertible matrices $A$ and $B$ such that $A\ve X$ and $B\ve Y$ are vectors of independent pairs,
exactly like in Lemma~\ref{thm:gensplit}. Therefore, we can use that theorem to give an explicit formula for the relaxed Wyner's common information between {\it arbitrarily correlated}  jointly Gaussian random vectors, as stated in the following theorem.

\begin{theorem} \label{thm-Gauss-vector}
Let $\ve X$ and $\ve Y$ be jointly Gaussian random vectors of length $n$ and covariance matrix $K_{(\ve X, \ve Y)}$.
Then,
\begin{align}
&C_{\gamma}({\ve X};{\ve Y})= \min_{\gamma_i: \sum_{i=1}^n \gamma_i = \gamma} \sum_{i=1}^n C_{\gamma_i}(X_i;Y_i),
\end{align}
where 
\begin{align}
C_{\gamma_i}(X_i;Y_i)=\frac{1}{2} \log^+ \frac{(1+\rho_i)(1-\sqrt{1- e^{-2 \gamma_i}})}{(1-\rho_i)(1+\sqrt{1- e^{-2 \gamma_i}})}
\end{align}
and $\rho_i$ (for $i=1,\dots,n$) are the singular values of $K_{\ve X}^{-1/2} K_{\ve X \ve Y} K_{\ve Y}^{-1/2},$
where $K_{\ve X}^{-1/2}$ and $K_{\ve Y}^{-1/2}$ are defined to mean that only the positive eigenvalues are inverted.
\end{theorem}

\begin{remark}
Note that we do not assume that $K_{\ve X}$ and $K_{\ve Y}$ are of full rank.
Moreover, note that the case where $\ve X$ and $\ve Y$ are of unequal length is included:
Simply invoke Lemma~\ref{Lemma-relWyner-basicprops}, Item 4), to append the shorter vector with independent Gaussians
so as to end up with two vectors of the same length.
\end{remark}

\begin{proof}
Note that the mean is irrelevant for the problem at hand, so we assume it to be zero without loss of generality.
The first step of the proof is to apply the same transform used, e.g., in~\cite{Satpathy--Cuff}.
Namely, we form $\hat{\ve{X}} = K_{\ve X}^{-1/2} {\ve{X}}$ and $\hat{\ve{Y}} = K_{\ve Y}^{-1/2} {\ve{Y}},$
where $K_{\ve X}^{-1/2}$ and $K_{\ve Y}^{-1/2}$ are defined to mean that only the positive eigenvalues are inverted.
Let us denote the rank of $K_{\ve X}$ by $r_X$ and the rank of $K_{\ve Y}$ by $r_Y.$
Then, we have
\begin{align}
K_{\hat{\ve X}}&= \left( \begin{array}{cc} I_{r_X} & 0 \\ 0 & 0_{n-r_X} \end{array} \right)
\end{align}
and
\begin{align}
K_{\hat{\ve Y}}&= \left( \begin{array}{cc} I_{r_Y} & 0 \\ 0 & 0_{n-r_Y} \end{array} \right)
\end{align}
Moreover, we have $K_{\hat{\ve X}\hat{\ve Y}}= K_{\ve X}^{-1/2} K_{\ve X \ve Y} K_{\ve Y}^{-1/2}$. Let us denote the singular value decomposition of this matrix by $K_{\hat{\ve X}\hat{\ve Y}}=R_{\ve X} \Lambda R_{\ve Y}$. Define $\tilde{\ve X} = R_{\ve X}^T \hat{\ve X}$ and $\tilde{\ve Y} = R_{\ve Y} \hat{\ve Y}$, which implies that $K_{\tilde{\ve X}}=K_{\hat{\ve X}},$ $K_{\tilde{\ve Y}}=K_{\hat{\ve Y}},$ and $K_{\tilde{\ve X}\tilde{\ve Y}}=\Lambda$.
The second step of the proof is to observe that the mappings from $\ve X$ to $\tilde{\ve X}$ and from $\ve Y$ to $\tilde{\ve Y},$ respectively, are linear one-to-one and mutual information is preserved under such transformation.  
%by Lemma~\ref{Lemma-relWyner-basicprops}, Item 6), 
Hence, we have $C_{\gamma}({\ve X};{\ve Y})=C_{\gamma}(\tilde{\ve X};\tilde{\ve Y}).$
The third, and key, step of the proof is now to observe that $\{ (X_i,Y_i) \}_{i=1}^n$ are $n$ independent pairs of random variables.
Hence, we can apply Lemma~\ref{thm:gensplit}.
The final step is to apply Theorem~\ref{thm:scalar} separately to each of the independent pairs, thus establishing the claimed formula.
\end{proof}

In the remainder of this section, we explore the structure of the allocation problem in Theorem~\ref{thm-Gauss-vector}, that is, the problem of optimally choosing the values of $\gamma_i.$
As we will show, the answer is of the water-filling type. That is, there is a ``water level'' $\gamma^*.$ Then, all $\gamma_i$ whose corresponding correlation coefficient $\rho_i$ is large enough will be set equal to $\gamma^*.$ The remaining $\gamma_i,$ corresponding to those $i$ with low correlation coefficient $\rho_i,$ will be set to their respective maximal values (all of which are smaller than $\gamma^*$).
To establish this result, we prefer to change notation as follows. We define $\alpha_i=\sqrt{1-e^{-2\gamma_i}}.$
With this, we can express the allocation problem in Theorem~\ref{thm-Gauss-vector} as
\begin{align}
C_{\gamma}(\ve X ; \ve Y)&= \min_{\alpha_1, \alpha_2, \cdots, \alpha_n}  \sum_{i=1}^n \frac{1}{2} \log^+ \frac{(1+\rho_i)(1-\alpha_i)}{(1-\rho_i)(1+\alpha_i)}  \mbox{ such that }   \sum_{i=1}^n \frac{1}{2} \log \frac{1}{1-\alpha_i^2} \le \gamma. \label{RevWF-Eqbasic}
\end{align}
Moreover, defining
\begin{align}
C(\rho) = \frac{1}{2} \log \frac{1+\rho}{1-\rho}, \quad I(\rho) = \frac{1}{2} \log \frac{1}{1-\rho^2},
\end{align}
we can rewrite Equation~\eqref{RevWF-Eqbasic} as
\begin{align} 
C_\gamma(\ve X ; \ve Y)&= \min_{\alpha_1, \alpha_2, \cdots, \alpha_n}  \sum_{i=1}^n \left( C(\rho_i)-C(\alpha_i) \right)^+  \mbox{ such that }   \sum_{i=1}^n I({\alpha_i}) \le \gamma. \label{RevWF-Eqbasic2}
\end{align}

\begin{theorem}\label{thm-waterfilling}
The solution to the allocation problem of Theorem~\ref{thm-Gauss-vector} can be expressed as
\begin{align}
C_\gamma(\ve X ; \ve Y)&= \sum_{i=1}^n \left( C({\rho_i})- \beta^*\right)^+ ,
\end{align}
where $\beta^*$ is selected such that
\begin{align}
 \sum_{i=1}^n \min \left\{ f(\beta^*), I({\rho_i}) \right\} &= \gamma,
\end{align}
where
\begin{align}
f(\beta^*) &= \frac{1}{2} \log \frac{(\exp(2\beta^*)+1)^2}{4\exp(2\beta^*)}.
\end{align}
\end{theorem}
%{\bf Note: This still needs a proof, but is obvious from our earlier discussions.}

%For illustration purposes, we can also write out a closed-form formula for the case $n=2,$ as follows.
%\begin{corollary}
%Assuming without loss of generality that $\rho_1>\rho_2,$ we have 
%\begin{align}
%&C_{\gamma}(\ve X;\ve Y)= \nonumber \\
%&
%\left\{ \begin{array}{lr}      
%      \frac{1}{2}  \log{\frac{(1+\rho_1)(1+\rho_2)(1-\sqrt{1-e^{-\gamma}})^2}{(1-\rho_1)(1-\rho_2)(1+\sqrt{1-e^{-\gamma}})^2}},  &     0 \le  \gamma < 2I(\rho_2), \\
%      \frac{1}{2} \log{\frac{(1+\rho_1)\left(1-\sqrt{1-\frac{e^{-2\gamma}}{1-\rho_2^2}}\right)}{(1-\rho_1)\left(1+\sqrt{1-\frac{e^{-2\gamma}}{1-\rho_2^2}}\right)}}, &2I(\rho_2) \le \gamma < I(\rho_1) + I(\rho_2)  , \\
% 0,  &   I(\rho_1) + I(\rho_2) \le \gamma . \\
%  \end{array} \right.
%\end{align}

%where
%\begin{align}
%   \overline{\gamma}_i &= \frac{1}{2} \log\frac{1}{1-\rho_i^2}.
%\end{align}
%\end{corollary}

\begin{proof}[Proof of Theorem~\ref{thm-waterfilling}]
Note that $(\ref{RevWF-Eqbasic2})$ can be rewritten as
%\begin{align}
%C_{\gamma}(\ve X ; \ve Y)&= \min_{\alpha_1, \alpha_2, \cdots, \alpha_n}  \sum_{i=1}^n \left( C({\rho_i})-C({\alpha_i}) \right)^+  \mbox{ such that }   \sum_{i=1}^n I({\alpha_i}) \le \gamma,
%\end{align}
\begin{align}
C_{\gamma}(\ve X ; \ve Y)&= \min_{\gamma_1, \gamma_2, \cdots, \gamma_n}  \sum_{i=1}^n \left( C({\rho_i})-C(I^{-1}(\gamma_i)) \right)^+  \mbox{ such that }   \sum_{i=1}^n \gamma_{i} \le \gamma, \label{eq-thm-waterfilling}
\end{align}
and thus, for notational compactness, let us define
\begin{align}
g(x) &= C(I^{-1}(x)) = \frac{1}{2}\log{\frac{1+\sqrt{1-e^{-2x}}}{1-\sqrt{1-e^{-2x}}}},
\end{align}
which is a strictly concave, strictly increasing function.
We also define its inverse,
\begin{align}
f(x)&= g^{-1}(x) = I (C^{-1}(x)) = \frac{1}{2} \log \frac{1}{1-\left(\frac{\exp(2x)-1}{\exp(2x)+1)}\right)^2} = \frac{1}{2} \log \frac{(\exp(2x)+1)^2}{4\exp(2x)},
\end{align}
which is a strictly convex, strictly increasing function.
%The proof only uses the fact that $C(I^{-1}(\gamma_i))$ is concave.

Without loss of generality, suppose that $\rho_1 \ge \rho_2 \ge \cdots \ge \rho_n.$ The objective function is composed of $n$ terms which can be active or not, meaning that they can be either positive or zero. Since the function $C(\rho)$ is increasing in $\rho,$ we have that $C(\rho_1) \geq C(\rho_2) \geq \cdots \geq C(\rho_n)$.
To summarize the intuition of the proof, note that the $n$-th term, {\it i.e.,} $\left( C({\rho_n})-g(\gamma_n) \right)^+,$ will be inactive first. Therefore, by increasing $\gamma$ then the terms will become inactive in a decreasing fashion until we are left with only the first term active and the rest inactive. 

Let us start with the case when they are all active, which means that $\sum_{i=1}^n \left( C({\rho_i})-g(\gamma_i) \right)^+ = \sum_{i=1}^n \left( C({\rho_i})-g(\gamma_i) \right)$
%Let us set $I_i = \lambda_i \gamma,$ where the $\lambda_i$ sum to one. 
Then, by the concavity of $g(\gamma_i),$ we have
\begin{align}
 \sum_{i=1}^n g(\gamma_i) &\le  n g(\frac{\gamma}{n}),
\end{align}
thus an optimal choice is $\gamma^* = \frac{\gamma}{n},$ for all $i.$ Hence, in our notation, in this case $\beta^* = g(\frac{\gamma}{n}).$ Clearly, all the terms are active in the interval $0\le \gamma\le n I({\rho_n}),$ with the reasoning that if the $n$-th terms is active then the rest of the terms is active too. Next, consider the case when the $n$-th term is inactive and the rest is active. Therefore, $\sum_{i=1}^n \left( C({\rho_i})-g(\gamma_i) \right)^+ = \sum_{i=1}^{n-1} \left( C({\rho_i})-g(\gamma_i) \right)$ and by the concavity of $g(\gamma_i),$ we have
\begin{align}
 \sum_{i=1}^{n-1} g(\gamma_i) &\le  (n-1) g \left(\frac{\gamma}{n-1} \right),
\end{align}
thus an optimal choice is $\gamma^* = \frac{\gamma-\gamma_n}{n-1},$ for all $i \in \{1,2,\cdots ,n-1\}$. The optimal choice for $\gamma_n$ is $\gamma_n=I(\rho_n)$, which makes the $n$-th term exactly zero. This scenario will happen in the interval, $n I({\rho_n}) < \gamma \le  I({\rho_n}) + (n-1) I({\rho_{n-1}}).$
%The above inequality is of course still valid, but we cannot set $I_M = \gamma/M$ since in the considered interval, we have $C(I=\gamma/M) > C_{\rho_M}.$
Instead, the corresponding $\beta^*$ in our notation is $\beta^* = g \left(\frac{\gamma-I({\rho_n})}{n-1} \right).$ In general, let us consider the case when $k$-th term is active and $k+1$-th is inactive. By a similar argument as above, the optimal choice is $\gamma^*=\frac{\gamma - \sum_{i=k+1}^n \gamma_i}{n-k}$ for $i \in \{ 1,2,\cdots,k\}$ and $\gamma_i=I(\rho_i)$ for $i \in \{ k+1, \cdots,n \}$. This scenario will happen in the interval $(k+1)I(\rho_{k+1}) +\sum_{i=k+2}^n I(\rho_i)< \gamma \leq kI(\rho_k) +\sum_{i=k+1}^n I(\rho_i)$. Importantly, observe that the optimal $\gamma_i$ can be rewritten as $\gamma_i=\min \{ I(\rho_i), \gamma^* \}$, therefore the solution to the allocation problem can be expressed as 
\begin{align}
C_\gamma(\ve X ; \ve Y)&= \sum_{i=1}^n \left( C({\rho_i})- g(\gamma^*)\right)^+ ,
\end{align}
where $\gamma^*$ is selected such that
\begin{align}
 \sum_{i=1}^n \min \left\{ \gamma^*, I({\rho_i}) \right\} &= \gamma.
\end{align}
The solution to the allocation problem can be rewritten as \begin{align}
C_\gamma(\ve X ; \ve Y)&= \sum_{i=1}^n \left( C({\rho_i})- \beta^*\right)^+ ,
\end{align}
where $\beta^*$ is selected such that
\begin{align}
 \sum_{i=1}^n \min \left\{ f(\beta^*), I({\rho_i}) \right\} &= \gamma.
\end{align}

%Perhaps better the other way round:
%we instead write
%\begin{align}
%C(\ve X ; \ve Y)&= \min_{\beta_1, \beta_2, \cdots, \beta_M}  \sum_{i=1}^n \left( C({\rho_i})- \beta_i \right)^+  \mbox{ such that }   \sum_{i=1}^n \min\{ I(C^{-1}(\beta_i)), I({\rho_i}) \} \le \gamma.
%\end{align}
%Since $I(C)$ is convex, it is best to pick them all equal.
%Okay, not sure this will be any better. To be discussed.
\end{proof}

Theorem~\ref{thm-waterfilling} shows that the allocation problem has a natural reverse water-filling interpretation which can be visualized in two dual ways.
First, we could consider the space of the $\gamma_i$ parameters, which leads to Figure~\ref{fig:bar}: None of the $\gamma_i$ should be selected larger than the corresponding $I(\rho_i),$ and those $\gamma_i$ that are strictly smaller than their maximum value should all be equal. This graphically identifies the optimal value $\gamma^*,$ and thus, the resulting solution to our optimization problem.
Alternatively, we could consider directly the space of the individual contributions to the objective, denoted by $C(\rho_i)$ in Equation~\eqref{eq-thm-waterfilling}, which leads to Figure~\ref{fig:bar2}.
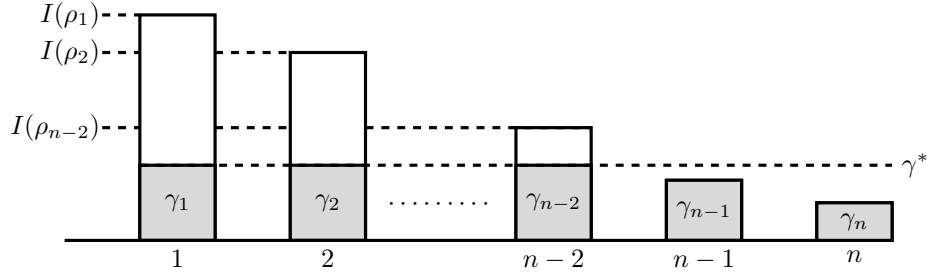
\begin{figure}[ht]% 
\centering
\begin{tikzpicture}
\draw[black,very thick] (1,0) rectangle (2,3);
\draw[black,very thick] (3,0) rectangle (4,2.5);
\draw[black,very thick] (6,0) rectangle (7,1.5);
\draw[black,very thick] (8,0) rectangle (9,0.8);
\draw[black,very thick] (10,0) rectangle (11,0.5);
\draw[black,thick,anchor=north] (1.5,0) node{$1$};
\draw[black,thick,anchor=north] (3.5,0) node{$2$};
\draw[black,thick,anchor=north] (6.5,0) node{$n-2$};
\draw[black,thick,anchor=north] (8.5,0) node{$n-1$};
\draw[black,thick,anchor=north] (10.5,0) node{$n$};
\draw[black,thick,anchor=east] (0.6,1.5) node{$I(\rho_{n-2})$};
\draw[black,thick,anchor=east] (0.6,2.5) node{$I(\rho_2)$};
\draw[black,thick,anchor=east] (0.6,3) node{$I(\rho_1)$};
\draw[black,thick,anchor=west] (11,1) node{$\gamma^*$};
\filldraw[color=black, fill=black!15, very thick](1,0) rectangle (2,1);
\filldraw[color=black, fill=black!15, very thick](3,0) rectangle (4,1);
\filldraw[color=black, fill=black!15, very thick](6,0) rectangle (7,1);
\filldraw[color=black, fill=black!15, very thick](8,0) rectangle (9,0.8);
\filldraw[color=black, fill=black!15, very thick](10,0) rectangle (11,0.5);
\draw[black,thick] (1.5,0.5) node{$\gamma_1$};
\draw[black,thick] (3.5,0.5) node{$\gamma_2$};
\draw[black,thick] (6.5,0.5) node{$\gamma_{n-2}$};
\draw[black,thick] (8.5,0.4) node{$\gamma_{n-1}$};
\draw[black,thick] (10.5,0.25) node{$\gamma_{n}$};
\draw[dashed,black,very thick](0.5,3)--(1,3);
\draw[dashed,black,very thick](0.5,2.5)--(1,2.5);
\draw[dashed,black,very thick](2,2.5)--(3,2.5);
\draw[dashed,black,very thick](0.5,1.5)--(1,1.5);
\draw[dashed,black,very thick](2,1.5)--(3,1.5);
\draw[dashed,black,very thick](4,1.5)--(7,1.5);
\draw[-,black,very thick](0,0)--(10,0);
\draw[dashed,black,very thick](1,1)--(11,1);
\draw[black,very thick,anchor=north] (5,0.7) node{$\cdots  \cdots \cdots$};
\end{tikzpicture}
\caption{Example of reverse water-filling. The (whole) bars represent the $\gamma_i$-s which make $C_{\gamma_i}(X_i;Y_i)=0$, and the shaded area of the bars is the proper allocation $\gamma_i$ to minimize the original problem. In this example, $\gamma=\sum_{i=1}^n \gamma_i$ is chosen such that $C_{\gamma_{n-1}}(X_{n-1};Y_{n-1})=C_{\gamma_{n}}(X_{n};Y_{n})=0$.} 
\label{fig:bar}
\end{figure}

\begin{figure}[ht]% 
\centering
\begin{tikzpicture}
\draw[black,very thick] (1,0) rectangle (2,3);
\draw[black,very thick] (3,0) rectangle (4,2.5);
\draw[black,very thick] (6,0) rectangle (7,1.5);
\draw[black,very thick] (8,0) rectangle (9,0.8);
\draw[black,very thick] (10,0) rectangle (11,0.5);
\draw[black,thick,anchor=north] (1.5,0) node{$1$};
\draw[black,thick,anchor=north] (3.5,0) node{$2$};
\draw[black,thick,anchor=north] (6.5,0) node{$n-2$};
\draw[black,thick,anchor=north] (8.5,0) node{$n-1$};
\draw[black,thick,anchor=north] (10.5,0) node{$n$};
\draw[black,thick,anchor=east] (0.6,1.5) node{$C(\rho_{n-2})$};
\draw[black,thick,anchor=east] (0.6,2.5) node{$C(\rho_2)$};
\draw[black,thick,anchor=east] (0.6,3) node{$C(\rho_1)$};
\draw[black,thick,anchor=west] (11,1) node{$\beta^*$};
\filldraw[color=black, fill=black!15, very thick](1,3) rectangle (2,1);
\filldraw[color=black, fill=black!15, very thick](3,2.5) rectangle (4,1);
\filldraw[color=black, fill=black!15, very thick](6,1.5) rectangle (7,1);
\draw[black,thick] (3.5,3) node{$C(\rho_1)-\beta^*$};
\draw[black,thick] (5.5,2) node{$C(\rho_2)-\beta^*$};
\draw[black,thick] (8.5,1.5) node{$C(\rho_{n-2})-\beta^*$};
\draw[->,thick, bend angle=45, bend left] (1.5,2) to (2.5,3);
\draw[->,thick, bend angle=45, bend left] (3.5,1.5) to (4.5,2);
\draw[->,thick, bend angle=45, bend left] (6.5,1.2) to (7.3,1.5);
\draw[dashed,black,very thick](0.5,3)--(1,3);
\draw[dashed,black,very thick](0.5,2.5)--(1,2.5);
\draw[dashed,black,very thick](2,2.5)--(3,2.5);
\draw[dashed,black,very thick](0.5,1.5)--(1,1.5);
\draw[dashed,black,very thick](2,1.5)--(3,1.5);
\draw[dashed,black,very thick](4,1.5)--(7,1.5);
\draw[-,black,very thick](0,0)--(10,0);
\draw[dashed,black,very thick](1,1)--(11,1);
\draw[black,very thick,anchor=north] (5,0.7) node{$\cdots  \cdots \cdots$};
\end{tikzpicture}
\caption{Example of reverse water-filling. The (whole) bars represent the (standard) Wyner's common information of each individual pair, respectively. The shaded area of the bars is the respective contribution to $C_\gamma(\ve X ; \ve Y).$ In this example, $\gamma$ is chosen such that $(C(\rho_{n-1}) -\beta^*)^+=(C(\rho_{n}) -\beta^*)^+=0$.} 
\label{fig:bar2}
\end{figure}
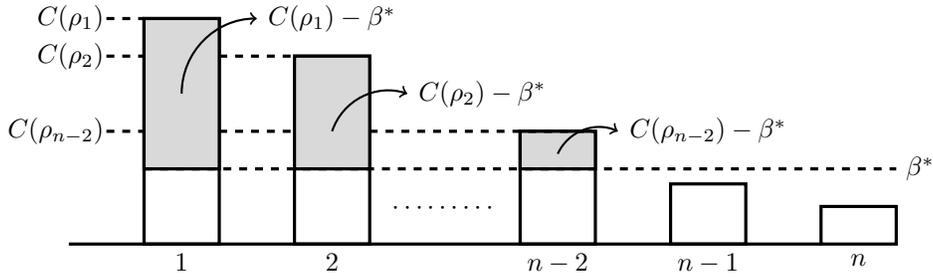

\section{The Gaussian Gray-Wyner Network}\label{Sec-GaussianGrayWyner}
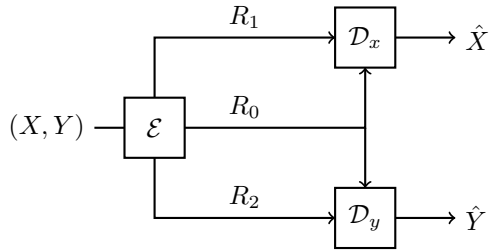
\begin{figure}[!ht]% 
\centering
\begin{tikzpicture}[scale=0.8]
\draw[black,thick] (0,1) rectangle (1,2);
\draw[black,thick,anchor=east] (-0.5,1.5) node{$(X,Y)$};
\draw[black,thick,anchor=west] (5.5,3) node{$\hat{X}$};
\draw[black,thick,anchor=west] (5.5,0) node{$\hat{Y}$};
\draw[black,thick,anchor=center] (0.5,1.5) node{$\mathcal{E}$};
\draw[black,thick,anchor=south] (2,0) node{$R_2$};
\draw[black,thick,anchor=south] (2,1.5) node{$R_0$};
\draw[black,thick,anchor=south] (2,3) node{$R_1$};
\draw[black,thick](-0.5,1.5)--(0,1.5);
\draw[->,black,thick](0.5,1)--(0.5,0)--(3.5,0);
\draw[->,black,thick](1,1.5)--(4,1.5)--(4,2.5);
\draw[->,black,thick](4,1.5)--(4,0.5);
\draw[->,black,thick](0.5,2)--(0.5,3)--(3.5,3);
\draw[black,thick,anchor=center] (4,0) node{$\mathcal{D}_y$};
\draw[black,thick,anchor=center] (4,3) node{$\mathcal{D}_x$};
\draw[black,thick] (3.5,-0.5) rectangle (4.5,0.5);
\draw[black,thick] (3.5,2.5) rectangle (4.5,3.5);
\draw[->,black,thick](4.5,0)--(5.5,0);
\draw[->,black,thick](4.5,3)--(5.5,3);
\end{tikzpicture}
\caption{The Gray-Wyner Network} 
\label{fig:Gray-Wyner}
\end{figure}

The Gray-Wyner network~\cite{Gray--Wyner} is composed of one sender and two receivers, as illustrated in Figure~\ref{fig:Gray-Wyner}. In a nutshell, the sender compresses two underlying correlated sources $X$ and $Y$ (with fixed $p(x,y)$) into three descriptions.
Here, we follow the notation and formal problem statement as given in~\cite[Section II]{Gray--Wyner}.
The central description, of rate $R_0,$ is provided to both receivers. Additionally, each receiver also has access to a tailored private description at rates $R_1$ and $R_2,$ respectively.
At the receivers, reconstruction is accomplished to within a fidelity criterion.
For given fidelity requirements $\Delta_1$ and $\Delta_2$ in the reconstruction of sources $X$ and $Y,$
respectively, we seek to characterize the set of achievable rate triples $(R_0, R_1, R_2),$
again following~\cite[Section II]{Gray--Wyner} to the letter.
The full solution, up to the optimization over an auxiliary, is characterized in~\cite[Theorem 8]{Gray--Wyner}, see~\cite[Equations (40a)-(40b)]{Gray--Wyner}.
Namely, define the regions
\begin{align}
{\cal R}^{(W)}(\Delta_1, \Delta_2) = \left\{ (R_0, R_1, R_2) : R_0 \ge I(X,Y; W), R_1 \ge R_{X|W}(\Delta_1), R_2 \ge R_{Y|W}(\Delta_2) \right\}. \label{eq-GW-general}
\end{align}
Here, $R_{X|W}(\cdot)$ and $R_{Y|W}(\cdot)$ denote the conditional rate-distortion functions of $X$ and $Y,$ respectively, given $W,$ see~\cite{Gray--1972}.
Then, the optimal region, denoted by ${\cal R}^*(\Delta_1, \Delta_2),$ is the (set) closure of the union of these regions over all choices of $W.$
The difficulty with this result is taking the union over all $W.$ 

For the jointly Gaussian source $(X,Y)$ subject to mean-squared error distortion, the complete solution remains unknown.
An account of this special case already appears in~\cite[Section 2.5(B)]{Gray--Wyner}.
Partial progress was made in \cite{Akyol--Rose--2014,Xu--Liu--Chen--2016} for the special case where $R_0+R_1+R_2 = R_{X,Y}(\Delta_1,\Delta_2).$
Here, $R_{X,Y}(\Delta_1,\Delta_2)$ denotes the rate-distortion function of {\it jointly} encoding $X$ and $Y$ to fidelities $\Delta_1$ and $\Delta_2,$ respectively.
By a simple cut-set argument, any scheme must satisfy $R_0+R_1+R_2 \ge R_{X,Y}(\Delta_1,\Delta_2),$
and thus, if a scheme attains this bound with equality, it is necessarily optimal.
It is immediate that if $R_0$ is large enough, then it is possible to meet with equality $R_0+R_1+R_2 = R_{X,Y}(\Delta_1,\Delta_2).$
However, to date, no progress has been reported for the general case where it is not possible to attain this cut-set bound with equality.

The main contribution of the present paper is a closed-form solution for the general case.
Specifically, our techniques allow us to establish that restricting the union over all $W$ to only jointly Gaussian auxiliaries is without loss of optimality.
To keep notation simple, we consider the following symmetric projection of the optimal rate region:
\begin{align}
\mathrm{R}_{\Delta, \alpha}(X, Y) &= \min R_0 \mbox{ such that }  (R_0, R_1, R_2) \in {\cal R}^*(\Delta, \Delta) \mbox{ and } R_1+R_2\le \alpha.
\end{align}
Using Equation~\eqref{eq-GW-general}, we can express $\mathrm{R}_{\Delta, \alpha}(X, Y)$ explicitly as the following optimization problem:
\begin{align}
\mathrm{R}_{\Delta, \alpha}(X, Y) &= \inf I(X,Y; W) \nonumber \\
 & \mbox{ such that }  I(X;\hat{X}|W) + I(Y;\hat{Y}|W) \le \alpha \mbox{ and } {\mathbb E}[(X-\hat{X})^2] \le \Delta \mbox{ and } {\mathbb E}[(Y-\hat{Y})^2] \le \Delta,
\end{align}
where the infimum is over all distributions $p(w, \hat{x}, \hat{y}|x,y).$

Then, we have the following theorem:

\begin{theorem}\label{thm-Gauss-lossymum}
Let $X$ and $Y$ be jointly Gaussian with mean zero, equal variance $\sigma^2,$ and with correlation coefficient $\rho.$
Let the distortion measure be mean-squared error. Then,
\begin{align}
\mathrm{R}_{\Delta, \alpha}(X, Y) = \left\{ \begin{array}{lr} \frac{1}{2} \log^{+}{\frac{1+\rho}{2\frac{\Delta}{\sigma^2}e^{\alpha}+\rho-1}}, &   \mbox{ if } \sigma^2(1-\rho) \le \Delta e^{\alpha} \le \sigma^2  \\
 \frac{1}{2} \log^{+}{\frac{1-\rho^2}{\frac{\Delta^2}{\sigma^4}e^{2\alpha}}} , &   \mbox{ if }  \Delta e^\alpha \le \sigma^2(1-\rho).
\end{array} \right. \nonumber
\end{align}
\end{theorem}

\begin{proof}
First, we observe that for mean-squared error, the source variance is irrelevant: A scheme attaining distortion $\Delta$ for sources of variance $\sigma^2$
is a scheme attaining distortion $\Delta/\sigma^2$ on unit-variance sources, and vice versa.
Therefore, for ease of notation, in the sequel, we assume that the sources are of unit variance.
Then, we can bound:
\begin{align}
\mathrm{R}_{\Delta, \alpha}(X, Y)&= \inf_{ \substack{W,\hat{X},\hat{Y}:I(X;\hat{X}|W)+I(Y;\hat{Y}|W) \leq \alpha\\ {\mathbb E}[(X-\hat{X})^2] \leq \Delta \\ {\mathbb E}[(Y-\hat{Y})^2] \leq \Delta }} I(X,Y;W)  \\
& \geq \inf_{ \substack{W,\hat{X},\hat{Y}: {\mathbb E}[(X-\hat{X})^2] \leq \Delta \\ {\mathbb E}[(Y-\hat{Y})^2] \leq \Delta }} I(X,Y;W) +\nu(I(X;\hat{X}|W)+I(Y;\hat{Y}|W) - \alpha) \label{eqn:weakduality} \\
& = \inf_{ \substack{W,\hat{X},\hat{Y}: {\mathbb E}[(X-\hat{X})^2] \leq \Delta \\ {\mathbb E}[(Y-\hat{Y})^2] \leq \Delta }} h(X,Y) - \nu \alpha +\nu (h(X|W)+h(Y|W))-h(X,Y|W) \nonumber \\
& \quad \quad  -\nu(h(X|W,\hat{X})+h(Y|W,\hat{Y})) \\
& \geq   h(X,Y) - \nu \alpha + \nu \inf_W h(X|W)+h(Y|W)- \frac{1}{\nu}h(X,Y|W) \nonumber \\
& \quad \quad  + \inf_{ \substack{W,\hat{X},\hat{Y}: {\mathbb E}[(X-\hat{X})^2] \leq \Delta \\ {\mathbb E}[(Y-\hat{Y})^2] \leq \Delta }} -\nu(h(X|W,\hat{X})+h(Y|W,\hat{Y})) \label{eqn:minsplit} \\
& \geq h(X,Y) - \nu \alpha + \nu \cdot \min_{0 \preceq K^{\prime} \preceq \begin{pmatrix} 1 & \rho \\ \rho &1  \end{pmatrix}} h(X^{\prime})+h(Y^{\prime})- \frac{1}{\nu}h(X^{\prime},Y^{\prime})  \nonumber \\
& \quad \quad  +\nu \cdot \left( \min_{ \substack{(W,\hat{X},\hat{Y}) \in \mathcal{P}_G: \\ {\mathbb E}[(X-\hat{X})^2] \leq \Delta }} -h(X|W,\hat{X})+  \min_{ \substack{(W,\hat{X},\hat{Y}) \in \mathcal{P}_G: \\ {\mathbb E}[(Y-\hat{Y})^2] \leq \Delta }}-h(Y|W,\hat{Y}) \right)   \label{eqn:twobound} \\
& = h(X,Y) - \nu \alpha -\nu \log{(2 \pi e \Delta)}  + \nu \cdot \min_{0 \preceq K^{\prime} \preceq \begin{pmatrix} 1 & \rho \\ \rho &1  \end{pmatrix}} h(X^{\prime})+h(Y^{\prime})- \frac{1}{\nu}h(X^{\prime},Y^{\prime}) \label{eqn:twoboundeval} \\
& = \frac{1}{2} \log{(2\pi e)^2(1-\rho^2)} - \nu \alpha -\nu \log{(2 \pi e \Delta)}  + \frac{\nu}{2} \log{\frac{\nu^2}{2\nu-1}}-\frac{1-\nu}{2} \log{(2\pi e)^2\frac{(1-\rho)^2}{2\nu-1}} \label{eqn:hypereval2} \\
& = \left\{ \begin{array}{lr} \frac{1}{2} \log^{+}{\frac{1+\rho}{2\Delta e^{\alpha}+\rho-1}}, &   \mbox{ if } 1-\rho \le \Delta e^{\alpha} \le 1  \\
 \frac{1}{2} \log^{+}{\frac{1-\rho^2}{\Delta^2e^{2\alpha}}} , &   \mbox{ if }  \Delta e^\alpha \le 1-\rho.
\end{array} \right. \label{eqn:lastGrayWyner}
\end{align}
where~\eqref{eqn:weakduality} follows from weak duality for $ \nu \geq 0;$ \eqref{eqn:minsplit} follows from bounding the infimum of the sum with the sum of the infima of its summands, and the fact that relaxing the constraints cannot increase the value of the infimum;
\eqref{eqn:twobound} follows from Theorem~\ref{Thm:Hypercontract} where $\nu:=\frac{1}{1+\lambda}$ and for the constraint $0 \leq \lambda<1$ (indeed we can also include zero) to be satisfied we need $ \frac{1}{2} <\nu \leq1$ and \cite[Lemma~1]{Thomas} on each of the terms;
\eqref{eqn:twoboundeval} follows by observing
\begin{align}
h(X|W,\hat{X}) & = h(X-\hat{X}|W,\hat{X}) \\
  & \le h(X-\hat{X}) \\
  & \le \frac{1}{2} \log (2 \pi e \Delta),
\end{align}
where the last step is due to the fact that ${\mathbb E}[(X-\hat{X})^2] \leq \Delta$; (\ref{eqn:hypereval2}) follows from Lemma \ref{lem:lemmasymmery} for $\nu \geq \frac{1}{1+\rho}$; and (\ref{eqn:lastGrayWyner}) follows from maximizing 
\begin{align}
\ell(\nu):=\frac{1}{2} \log{(2\pi e)^2(1-\rho^2)} - \nu \alpha -\nu \log{(2 \pi e \Delta)}  + \frac{\nu}{2} \log{\frac{\nu^2}{2\nu-1}}-\frac{1-\nu}{2} \log{(2\pi e)^2\frac{(1-\rho)^2}{2\nu-1}},
\end{align}
for $1 \geq \nu \geq \frac{1}{1+\rho}$. Now we need to choose the tightest bound $\max_{1 \geq \nu \geq \frac{1}{1+\rho}} \ell(\nu)$. Note that the function $\ell$ is concave since
\begin{align}
\frac{\partial^2 \ell}{\partial \nu^2}=-\frac{1}{\nu(2\nu-1)}<0.
\end{align}
Since it also satisfies monotonicity
\begin{align}
\frac{\partial \ell}{\partial \nu}=\log{\frac{\nu(1-\rho)}{(2\nu-1)\Delta e^{\alpha}}},
\end{align}
its maximal value occurs when the derivative vanishes, that is, when $\nu_*=\frac{\Delta e^{\alpha}}{2\Delta e^{\alpha}-1+\rho}.$ Substituting for the optimal $\nu_*$ we get
\begin{align}
\mathrm{R}_{\Delta, \alpha}(X, Y) \geq \ell \left(\frac{\Delta e^{\alpha}}{2\Delta e^{\alpha}-1+\rho} \right) =\frac{1}{2} \log^+ \frac{1+\rho}{2 \Delta e^{\alpha}-1+\rho},
\end{align}
for $1 \geq \nu_* \geq \frac{1}{1+\rho}$, which means the expression is valid for $1-\rho \leq \Delta e^{\alpha} \leq 1$. 

The other case is $ \Delta e^{\alpha} \leq 1-\rho$. In this case note that $\nu(1-\rho) \geq \nu \Delta e^{\alpha} \geq (2\nu-1)\Delta e^\alpha$ for $\nu \leq 1$. This implies $\frac{\nu(1-\rho)}{(2\nu-1)\Delta e^{\alpha}} \geq 1$, thus we have $\frac{\partial \ell}{\partial \nu} \geq 0$. Since the function is concave and increasing the maximum is attained at $\nu_*=1$, thus 
\begin{align}
\mathrm{R}_{\Delta , \alpha}(X, Y) \geq \ell \left(1 \right) =\frac{1}{2} \log^+ \frac{1-\rho^2}{\Delta^2e^{2\alpha}},
\end{align}
where the expression is valid for $\Delta e^{\alpha} \leq 1-\rho$.
As stated at the beginning of the proof, this is the correct formula assuming unit-variance sources.
For sources of variance $\sigma^2,$ it suffices to replace $\Delta$ with $\Delta/\sigma^2,$ which leads to the expression given in the theorem statement.
%from proposition (\ref{prop:maingaussmum}), where $-h(X|W,\hat{X})=I(X;\hat{X}|W)-h(X|W)$ and when evaluated with the optimal distribution we get $I(X;\hat{X}|W)-h(X|W)=\frac{1}{2}\log{\frac{\Var(X|W)}{D}}-\frac{1}{2}\log{(2\pi e \Var(X|W))}$ and (\ref{eqn:lastGrayWyner}) is a consequence of the following lemma \ref{lemma:GrayWynerComb}.
\end{proof}

\section{Concluding Remarks}
%\section{Concluding Remarks and Open Problems}

We studied a natural relaxation of Wyner's common information, whereby the constraint of conditional independence is replaced by an upper bound on the conditional mutual information. This leads to a novel and different optimization problem.
We established a number of properties of this novel quantity, including a chain rule type formula for the case of independent pairs of random variables.
For the case of jointly Gaussian sources, both scalar and vector, we presented a closed-form expression for the relaxed Wyner's common information.
Finally, using the same tool set, we fully characterize the lossy Gaussian Gray-Wyner network subject to mean-squared error.
%Open problems include:
%\begin{itemize}
%\item The full solution for the binary symmetric source.% see Example~\ref{example-binary-RWCI}.
%\item For the case of Wyner's common information, Witsenhausen~\cite{Witsenhausen} managed to give closed-form formulas for a class of distributions he refers to as ``L-shaped.''
%An analogous investigation could be undertaken for the case of the relaxed Wyner's common information.
%\item An extension of the chain rule of Theorem~\ref{thm:gensplit} to other probabilistic models, beyond independent pairs.
%\end{itemize}

\appendices
%%%%%%%%%%%%%%%%%%%%%%%%%%%%%%%%%%%%%%%%%%%%%%%%%%%%%%%%%%%%%%%%%%%%%%%%
%%%%%%%%%%%%%%%%%%%%%%%%%%%%%%%%%%%%%%%%%%%%%%%%%%%%%%%%%%%%%%%%%%%%%%%%
%%%%%%%%%%%%%%%%%%%%%%%%%%%%%%%%%%%%%%%%%%%%%%%%%%%%%%%%%%%%%%%%%%%%%%%%
%\section{Proofs of Lemmas~\ref{Lemma-Wyner-basicprops} and~\ref{Lemma-relWyner-basicprops}}
\section{Proof of Lemma \ref{Lemma-relWyner-basicprops}} \label{app:Lemma-relWyner-basicprops}
%%%%%%%%%%%%%%%%%%%%%%%%%%%%%%%%%%%%%%%%%%%%%%%%%%%%%%%%%%%%%%%%%%%%%%%%
%%%%%%%%%%%%%%%%%%%%%%%%%%%%%%%%%%%%%%%%%%%%%%%%%%%%%%%%%%%%%%%%%%%%%%%%

%\subsection{Proof of Lemma \ref{Lemma-Wyner-basicprops}} \label{app:Lemma-Wyner-basicprops}
%
%
%Items 1)-2) are proved in~\cite{Wyner}.
%For Item 2), $I(X;Y) \le I(X; Y, W) = I(X;W) \le I(X,Y ; W)$ for any $W$ under which $X$ and $Y$ are conditionally independent. The equality condition $C(X;Y) = I(X;Y)$ holds if and only if $C(X;Y)$ is also equal to the G\'{a}c-K\"{o}rner common information. 
%
%Item 3) is stated in~\cite{Witsenhausen:76}. To prove it, observe that for fixed $p(x,y,z),$ we can write
%\begin{align}
%C(X;Y) & =  \inf_{p(x,y,z)p(w|x,y): I(X;Y|W)=0} I(X,Y; W) \\
%   & \ge \inf_{p(x,y,z)p(w|x,y): I(X;Y|W)=0} I(X,Z; W),
%\end{align}
%due to the Markov chain $(X,Z) - (X,Y) - W.$
%Moreover, note that since we consider only joint distributions of the form $p(x,y)p(z|y)p(w|x,y),$
%we also have the Markov chain $(X,W)-Y-Z$, which implies the Markov chain $X-(W,Y)-Z$. The latter implies $I(X;Y|W)\ge I(X;Z|W).$
%Hence,
%\begin{align}
%C(X;Y) &  \ge \inf_{p(x,y,z)p(w|x,y): I(X;Z|W)=0} I(X,Z; W) \ge C(X;Z).
%\end{align}
%By the same token, $C(Y;Z)\ge C(X;Z),$ which completes the proof.

%\subsection{Proof of Lemma \ref{Lemma-relWyner-basicprops}} \label{app:Lemma-relWyner-basicprops}

%Item 1) is a standard cardinality bound, following from the arguments in~\cite{AhlswedeK:75}. For the context at hand, see also Theorem 1 in~\cite[p.6396]{7557078}.
For Item 1), the inequality follows from the fact that mutual information is non-negative.
If $\gamma \ge I(X;Y),$ we may select $W$ to be a constant, thus we have equality to zero.
If $\gamma < I(X;Y),$ then the lower bound proved in the next item establishes that we cannot have equality to zero.
Also, observe that the Lagrangian for the relaxed Wyner's common information problem of Equation~\eqref{Eq-def-Wyner-relaxed} is $L(\lambda, p(w|x,y)) = I(X,Y;W) + \lambda(I(X;Y|W)-\gamma).$
From Lagrange duality, we thus have the lower bound $C_{\gamma}(X;Y) \ge \inf_{p(w|x,y)} L(\lambda, p(w|x,y)),$ for all positive $\lambda.$
Setting $\lambda=1,$ we have $\inf_{p(w|x,y)} (I(X,Y;W) + I(X;Y|W) -\gamma) = \inf_{p(w|x,y)} (I(X;Y) + I(X;W|Y) + I(Y;W|X) -\gamma) = I(X;Y)-\gamma.$
For Item 2), observe that for fixed $p(x,y,z),$ we can write
\begin{align}
C_\gamma(X;Y) & =  \inf_{p(x,y,z)p(w|x,y): I(X;Y|W)\le \gamma} I(X,Y; W) \\
   & \ge \inf_{p(x,y,z)p(w|x,y): I(X;Y|W)\le \gamma} I(X,Z; W),
\end{align}
due to the Markov chain $(X,Z) - (X,Y) - W.$
Moreover, note that since we consider only joint distributions of the form $p(x,y)p(z|y)p(w|x,y),$
we also have the Markov chain $(X,W)-Y-Z$, which implies the Markov chain $X-(W,Y)-Z$. The latter implies $I(X;Y|W)\ge I(X;Z|W).$
Hence,
\begin{align}
C_\gamma(X;Y) &  \ge \inf_{p(x,y,z)p(w|x,y): I(X;Z|W)\le\gamma} I(X,Z; W) \ge C_\gamma(X;Z).
\end{align}
By the same token, $C_\gamma(Y;Z)\ge C_\gamma(X;Z),$ which completes the proof.
Item 3) follows directly from~\cite[Corollary 4.5]{1055346}.
For Item 4), on the one hand, we have
\begin{align}
C_\gamma((X,Z);Y) &=\inf_{p(w|x,y,z): I(X,Z; Y| W) \leq \gamma} I(X,Z,Y; W) \\
&\leq \inf_{p(w|x,y): I(X;Y| W)+I(Z; Y| W,X) \leq \gamma} I(X,Y; W) +I(Z;W|X,Y) \label{eqn:optimizeoverZupper} \\
&=C_{\gamma} (X;Y)
\end{align}
where in Equation~\eqref{eqn:optimizeoverZupper} we add the constraint that $W$ is selected to be {\it independent} of $Z,$
which cannot reduce the value of the infimum.
Clearly, for such a choice of $W,$ we have $I(Z;Y| W,X)=0$ and $I(Z; W|X,Y)=0,$ which thus establishes the last step.
Conversely, observe that
\begin{align}
C_\gamma((X,Z);Y) &=\inf_{p(w|x,y,z): I(X,Z; Y| W) \leq \gamma} I(X,Y,Z; W) \\
& \geq \inf_{p(w|x,y): I(X;Y| W) \leq \gamma} I(X,Y; W) + \inf_{p(w|x,y,z):  I(X,Z; Y| W) \leq \gamma} I(Z;W|X,Y) \label{eqn:splitsumrelaxLower} \\
& \geq C_{\gamma} (X;Y) \label{eqn:nonnegativeLower}
\end{align}
where $(\ref{eqn:splitsumrelaxLower})$ follows from the fact that the infimum of the sum is lower bounded by the sum of the infimums and the fact that relaxing constraints cannot increase the value of the infimum, and $(\ref{eqn:nonnegativeLower})$ follows from non-negativity of the second term. 
%Item 6) follows because all involved mutual information terms are invariant to one-to-one transforms.
%For Item 7), note that we can express $C_{\gamma} (X; X) = H(X) - \max_{p(w|x): H(X|W) \le \gamma} H(X|W),$ which directly gives the result.

%%%%%%%%%%%%%%%%%%%%%%%%%%%%%%%%%%%%%%%%%%%%%%%%%%%%%%%%%%%%%%%%%%%%%%%%
%%%%%%%%%%%%%%%%%%%%%%%%%%%%%%%%%%%%%%%%%%%%%%%%%%%%%%%%%%%%%%%%%%%%%%%%
\section{Proof of Lemma~\ref{thm:gensplit}} \label{app:chainsplit}
%%%%%%%%%%%%%%%%%%%%%%%%%%%%%%%%%%%%%%%%%%%%%%%%%%%%%%%%%%%%%%%%%%%%%%%%
%%%%%%%%%%%%%%%%%%%%%%%%%%%%%%%%%%%%%%%%%%%%%%%%%%%%%%%%%%%%%%%%%%%%%%%%
The achievability part, that is, the inequality
\begin{align}
C_{\gamma} (X^n; Y^n) &\le  \min_{\{\gamma_i\}_{i=1}^n : \sum_{i=1}^n \gamma_i=\gamma}  \sum_{i=1}^n C_{\gamma_i}(X_i; Y_i),
\end{align}
merely corresponds to a particular choice of $W$ in the definition given in Equation~\eqref{Eq-def-Wyner-relaxed}. Specifically, let $W=(W_1, W_2, \ldots, W_n)$, and choose $\{(X_i,Y_i,W_i)\}_{i=1}^n$ to be $n$ independent triples of random vectors.
The converse is more subtle. We prove the case $n=2$ first, followed by induction. For $n=2,$ we have
\begin{align}
\lefteqn{\inf_{p(w|x_1,x_2,y_1,y_2): I(X_1,X_2; Y_1,Y_2|W)\le \gamma} I(X_1,X_2,Y_1,Y_2;W)} \nonumber \\
 &\overset{(a)}{\ge} \inf_{p(w|x_1,x_2,y_1,y_2): I(X_1; Y_1|W) + I(X_2; Y_2|W,X_1)\le \gamma} I(X_1,Y_1;W) +  I(X_2,Y_2;W,X_1) \\
&\overset{(b)}{=} \min_{\gamma_1+\gamma_2 = \gamma} \left\{ \inf_{p(w|x_1,x_2,y_1,y_2): I(X_1; Y_1|W)\le \gamma_1,  I(X_2; Y_2|W,X_1)\le \gamma_2} I(X_1,Y_1;W) +  I(X_2,Y_2;W,X_1)  \right\} \\
&\overset{(c)}{\ge} \min_{\gamma_1+\gamma_2 = \gamma} \left\{   \inf_{p(w|x_1,x_2,y_1,y_2): I(X_1; Y_1|W)\le \gamma_1,  I(X_2; Y_2|W,X_1)\le \gamma_2} I(X_1,Y_1;W) \right. \\
 & \quad \left.+ \inf_{p(\tilde{w}|x_1,x_2,y_1,y_2): I(X_1; Y_1|\tilde{W})\le \gamma_1,  I(X_2; Y_2|\tilde{W},X_1)\le \gamma_2}  I(X_2,Y_2;\tilde{W},X_1) \right\} \\
 &\overset{(d)}{\ge} \min_{\gamma_1+\gamma_2 = \gamma} \left\{   \inf_{p(w|x_1,y_1): I(X_1; Y_1|W)\le \gamma_1} I(X_1,Y_1;W) + \inf_{p(\tilde{w}|x_1,x_2,y_2):  I(X_2; Y_2|\tilde{W},X_1)\le \gamma_2}  I(X_2,Y_2;\tilde{W},X_1) \right\} \\
 %&\overset{(e)}{=} \min_{\gamma_1+\gamma_2 \le \gamma} \left\{   \min_{p(w_1|x_1,y_1): I(X_1; Y_1|W_1)\le \gamma_1} I(X_1,Y_1;W_1) + \min_{p(\tilde{w}|x_2,y_2,x_1):  I(X_2; Y_2|\tilde{W},X_1)\le \gamma_2}  I(X_2,Y_2;\tilde{W},X_1) \right\} \\
 &\overset{(e)}{\ge} \min_{\gamma_1+\gamma_2 = \gamma} \left\{   \inf_{p(w_1|x_1,y_1): I(X_1; Y_1|W_1)\le \gamma_1} I(X_1,Y_1;W_1) + \inf_{p(\tilde{w},\tilde{x}_1|x_2,y_2):  I(X_2; Y_2|\tilde{W},\tilde{X}_1)\le \gamma_2}  I(X_2,Y_2;\tilde{W},\tilde{X}_1) \right\}
 \end{align}
where Step $(a)$ follows from 
\begin{align}
I(X_1,X_2,Y_1,Y_2;W)&=I(X_1,Y_1;W)+ I(X_2,Y_2;W|X_1,Y_1) +I(X_1,Y_1;X_2,Y_2)\\
&=I(X_1,Y_1;W)+ I(X_2,Y_2;W, X_1,Y_1) \\
& \geq  I(X_1,Y_1;W)+ I(X_2,Y_2;W, X_1)
\end{align}
and the constraint is relaxed as follows
\begin{align}
\gamma \geq I(X_1,X_2;Y_1,Y_2|W)&=I(X_1;Y_1,Y_2|W)+ I(X_2;Y_1,Y_2|W,X_1) \\
&\geq I(X_1;Y_1|W)+ I(X_2;Y_2|W, X_1),
\end{align}
Step $(b)$ follows from splitting the minimization, Step $(c)$ follows from minimizing each subproblem individually which would result in a lower bound to the original problem, Step $(d)$ follows from reducing the number of constraints resulting into a lower bound, Step $(e)$ follows from introducing $\tilde{X}_1$ as a random variable to be optimized, whereas before $X_1$ had a fixed distribution. In other words, the preceding minimization is taken over $p(\tilde{w}|x_2,y_2,x_1)p(x_1|x_2,y_2)$ where $p(x_1|x_2,y_2)$ has a fixed distribution, whereas now the minimization is taken over $p(\tilde{w}|x_2,y_2,\tilde{x}_1)p(\tilde{x}_1|x_2,y_2)$, where we also optimize over $p(\tilde{x}_1|x_2,y_2)$. Lastly, denoting $W_2 = (\tilde{W},\tilde{X}_1),$ this can be expressed as
\begin{align}
\lefteqn{\inf_{p(w|x_1,x_2,y_1,y_2): I(X_1,X_2; Y_1,Y_2|W)\le \gamma} I(X_1,X_2,Y_1,Y_2;W)} \nonumber \\
 & \ge \min_{\gamma_1+\gamma_2 = \gamma} \{   \inf_{p(w_1|x_1,y_1): I(X_1; Y_1|W_1)\le \gamma_1} I(X_1,Y_1;W_1) + \inf_{p(w_2|x_2,y_2):  I(X_2; Y_2|W_2)\le \gamma_2}  I(X_2,Y_2;W_2) \}
\end{align}
After proving it for $n=2$, we will use the standard induction. In other words, we will assume that the converse holds for $n-1$ i.e. 
\begin{align} \label{eqn:assumpnminusone}
C_{\bar{\gamma}} (X^{n-1}; Y^{n-1}) &\geq  \min_{\gamma_i : \sum_{i=1}^{n-1} \gamma_i=\bar{\gamma}}  \sum_{i=1}^{n-1} C_{\gamma_i}(X_i; Y_i),
\end{align}
after we prove it for $n$ as follows,
\begin{align}
\lefteqn{\inf_{p(w|x^n,y^n): I(X^n; Y^n|W)\le \gamma} I(X^n,Y^n;W)} \nonumber \\
 &\overset{(f)}{\ge} \inf_{p(w|x^n,y^n): I(X^{n-1}; Y^{n-1}|W) + I(X_n; Y_n|W,X^{n-1})\le \gamma} I(X^{n-1},Y^{n-1};W) +  I(X_n,Y_n;W,X^{n-1}) \\
&\overset{(g)}{=} \min_{\bar{\gamma}+\gamma_n = \gamma} \left\{ \inf_{p(w|x^n,y^n): \substack{ I(X^{n-1}; Y^{n-1}|W)\leq \sum_{i=1}^{n-1} \gamma_i, \\ I(X_n; Y_n|W,X^{n-1})\le \gamma_n}} I(X^{n-1},Y^{n-1};W) +  I(X_n,Y_n;W,X^{n-1})  \right\} \\
&\overset{(h)}{\ge} \min_{\bar{\gamma}+\gamma_n = \gamma} \left\{ \inf_{p(w|x^n,y^n): \substack{ I(X^{n-1}; Y^{n-1}|W)\leq \sum_{i=1}^{n-1} \gamma_i, \\ I(X_n; Y_n|W,X^{n-1})\le \gamma_n}} I(X^{n-1},Y^{n-1};W) \right. \\
 & \quad \left.+ \inf_{p(\tilde{w}|x^n,y^n): \substack{ I(X^{n-1}; Y^{n-1}|\tilde{W})\leq \sum_{i=1}^{n-1} \gamma_i, \\ I(X_n; Y_n|\tilde{W},X^{n-1})\le \gamma_n}}  I(X_n,Y_n;\tilde{W},X^{n-1}) \right\} \\
 &\overset{(i)}{\ge} \min_{\bar{\gamma}+\gamma_n = \gamma} \left\{ \inf_{p(w|x^{n-1},y^{n-1}): I(X^{n-1}; Y^{n-1}|W)\leq \sum_{i=1}^{n-1} \gamma_i} I(X^{n-1},Y^{n-1};W) \right. \\
 & \quad \left. + \inf_{p(\tilde{w}|x^n,y_n):  I(X_n; Y_n|\tilde{W},X^{n-1})\le \gamma_n}  I(X_n,Y_n;\tilde{W},X^{n-1}) \right\} \\
 &\overset{(j)}{\ge} \min_{\bar{\gamma}+\gamma_n = \gamma} \left\{ \inf_{p(w|x^{n-1},y^{n-1}): I(X^{n-1}; Y^{n-1}|W)\leq \sum_{i=1}^{n-1} \gamma_i} I(X^{n-1},Y^{n-1};W) \right. \\
 & \quad \left. + \inf_{p(\tilde{w},\tilde{x}^{n-1}|x_n,y_n):  I(X_n; Y_n|\tilde{W},\tilde{X}^{n-1})\le \gamma_n}  I(X_n,Y_n;\tilde{W},\tilde{X}^{n-1}) \right\} \\
 & \overset{(k)}{=} \min_{\bar{\gamma}+\gamma_n = \gamma} \left\{ C_{\bar{\gamma}} I(X^{n-1};Y^{n-1}) + \inf_{p(w_n|x_n,y_n):  I(X_n; Y_n|W_n)\le \gamma_n}  I(X_n,Y_n;W_n)\right\} \\
 & \overset{(\ell)}{\geq} \min_{\bar{\gamma}+\gamma_n = \gamma} \left\{ C_{\gamma_n}(X_n;Y_n) + \min_{\gamma_i : \sum_{i=1}^{n-1} \gamma_i=\bar{\gamma}}  \sum_{i=1}^{n-1} C_{\gamma_i}(X_i; Y_i) \right\} \\
& = \min_{\gamma_i : \sum_{i=1}^{n} \gamma_i=\gamma}  \sum_{i=1}^{n} C_{\gamma_i}(X_i; Y_i),
\end{align} 
where Step $(f)$ follows from 
\begin{align}
I(X^n,Y^n;W)&=I(X^{n-1},Y^{n-1};W)+ I(X_n,Y_n;W|X^{n-1},Y^{n-1}) +I(X_n,Y_n;X^{n-1},Y^{n-1})\\
&=I(X^{n-1},Y^{n-1};W)+ I(X_n,Y_n;W,X^{n-1},Y^{n-1}) \\
& \geq  I(X^{n-1},Y^{n-1};W)+ I(X_n,Y_n;W,X^{n-1})
\end{align}
and the constraint is relaxed as follows
\begin{align}
\gamma \geq I(X^n;Y^n|W)&=I(X^{n-1};Y^n|W)+ I(X_n;Y^n|W,X^{n-1}) \\
&\geq I(X^{n-1};Y^{n-1}|W)+ I(X_n;Y_n|W,X^{n-1}).
\end{align}
Step $(g)$ follows from the same argument as $(b)$, Step $(h)$ follows from the same argument as $(c)$, Step $(i)$ follows follows from the same argument as $(d)$, Step $(j)$ follows from a similar argument as $(e)$, Step $(k)$ follows from denoting $W_n = (\tilde{W},\tilde{X}^{n-1}),$ and Step $(\ell)$ follows from the induction hypothesis (\ref{eqn:assumpnminusone}).

%%%%%%%%%%%%%%%%%%%%%%%%%%%%%%%%%%%%%%%%%%%%%%%%%%%%%%%%%%%%%%%%%%%%%%%%
%%%%%%%%%%%%%%%%%%%%%%%%%%%%%%%%%%%%%%%%%%%%%%%%%%%%%%%%%%%%%%%%%%%%%%%%
%\section{Derivation of Formula~\eqref{Eq-operational-optprobsol}}\label{App-der-Eq-operational-optprobsol}
%%%%%%%%%%%%%%%%%%%%%%%%%%%%%%%%%%%%%%%%%%%%%%%%%%%%%%%%%%%%%%%%%%%%%%%%
%%%%%%%%%%%%%%%%%%%%%%%%%%%%%%%%%%%%%%%%%%%%%%%%%%%%%%%%%%%%%%%%%%%%%%%%

%Consider the optimization problem from Equation~\eqref{Eq-operational-optprob}. Using Equation~\eqref{eq-GrayWyner}, it can be expressed as
%\begin{align}
%{R}_u^* &= \min  H(X|W) + H(Y|W)  \mbox { such that }  I(X,Y; W) \le \delta.
%\end{align}
%Equivalently, we can write
%\begin{align}
%{R}_u^* &= \min H(X,Y) - H(X,Y) + H(X,Y|W) - H(X,Y|W) +  H(X|W) + H(Y|W)  \mbox { such that }  I(X,Y; W) \le \delta.
%\end{align}
%Collecting terms, this can be expressed as
%\begin{align}
%{R}_u^* &= \min H(X,Y) - I(X,Y; W) + I(X;Y|W)  \mbox { such that }  I(X,Y; W) \le \delta.
%\end{align}
%As long as $\delta \le H(X,Y),$ this can be rewritten as
%\begin{align}
%{R}_u^* &=  H(X,Y) - \delta + \min I(X;Y|W)  \mbox { such that }  I(X,Y; W) \le \delta,
%\end{align}
%which can now be rewritten as
%\begin{align}
%{R}_u^* &=  H(X,Y) - \delta + C^{-1}_\delta(X;Y).
%\end{align}
%where
%\begin{align}
% C^{-1}_\delta(X;Y) & = \min \{ \gamma : C_\gamma(X;Y) \le \delta \}.
%\end{align}

%%%%%%%%%%%%%%%%%%%%%%%%%%%%%%%%%%%%%%%%%%%%%%%%%%%%%%%%%%%%%%%%%%%%%%%%
%%%%%%%%%%%%%%%%%%%%%%%%%%%%%%%%%%%%%%%%%%%%%%%%%%%%%%%%%%%%%%%%%%%%%%%%
\section{Proof of Theorem \ref{Thm:Hypercontract}} \label{app:HyperGauss}
The techniques to establish the optimality of Gaussian distributions is used in \cite{Geng--Nair} and is known as factorization of lower convex envelope. Let us define the following object
\begin{align} \label{eqn:convexenv_object}
V(K)=\inf_{(X,Y):K_{(X,Y)}=K} \inf_W h(Y|W)+h(X|W) - (1+\lambda) h(X, Y|W),
\end{align}
where $\lambda$ is a real number, $0<\lambda<1$  and $K$ is an arbitrary covariance matrix. Let $\ell_{\lambda}(X,Y)=h(Y)+h(X) - (1+\lambda) h(X,Y)$, and $\breve{\ell}_{\lambda}(X,Y) =\inf_W h(Y|W)+h(X|W) - (1+\lambda) h(X, Y|W)$, where $\breve{\ell}_{\lambda}(X,Y)$ is the lower convex envelope of $\ell_{\lambda}(X,Y)$. 
%Thus, $V(K)$ can be rewritten as
%\begin{align} \label{eqn:convexenv_meetsfunc}
%V(K)=\inf_{(X,Y):K_{(X,Y)}=K} h(Y)+h(X) - (1+\lambda) h(X,Y). 
%\end{align}
%Assuming $\mathbb{E}[W_n]=0$ and $\mathbb{E}[W_n^2] < \infty$ for all $n$, will guarantee that the sequence of random variables $\left\{W_n\right\}|(X,Y)=(x,y)$ has a finite variance. 

First, in Section~\ref{app:HyperGauss:infattained}, we prove that the infimum is attained, then, in Section~\ref{app:HyperGauss:Gauss}, we prove that a Gaussian $W$ attains the infimum in Equation~\eqref{eqn:convexenv_object}. Together, these two arguments establish Theorem \ref{Thm:Hypercontract}.

%%%%%%
\subsection{The infimum in Equation~\eqref{eqn:convexenv_object} is attained}\label{app:HyperGauss:infattained}
\begin{proposition}[Proposition 17 in \cite{Geng--Nair}]
Consider a sequence of random variables $\left\{X_n,Y_n\right\}$ such that $K_{(X_n,Y_n)}\preceq K$ for all $n$, then the sequence is tight.
\end{proposition}

\begin{theorem}[Prokhorov] \label{thm:weakconvergence}
If $\left\{X_n,Y_n\right\}$ is a tight sequence then there exists a subsequence $\left\{X_{n_i},Y_{n_i}\right\}$ and a limiting probability distribution $\left\{X_*,Y_*\right\}$ such that $\left\{X_{n_i},Y_{n_i}\right\} \overset{w}{\Rightarrow} \left\{X_*,Y_*\right\}$ converges weakly in distribution.
\end{theorem} 

Note that $\ell_{\lambda}(X,Y)=h(Y)+h(X) - (1+\lambda) h(X,Y)$ can be written as $(1+\lambda)I(X;Y)-\lambda[h(X)+h(Y)]$. Thus, it is enough to show that this expression is lower semi-continuous. We will show by utilizing the following theorem.

\begin{theorem}[\cite{Posner}]
If $p_{X_n,Y_n} \overset{w}{\Rightarrow} p_{X,Y}$ and $q_{X_n,Y_n} \overset{w}{\Rightarrow} q_{X,Y}$, then $D(p_{X,Y}||q_{X,Y}) \leq \liminf\limits_{n \to \infty} D(p_{X_n,Y_n}||q_{X_n,Y_n})$.
\end{theorem}

%Note that $I(XY;W)=D(P_{XYW}||Q_{XY}Q_W)$, therefore $I(XY;W) \leq \lim\inf\limits_{n \to \infty} I(X_nY_n;W_n)$, since $(X_n,Y_n,W_n) \overset{w}{\Rightarrow} (X,Y,W)$, $(X_n,Y_n) \overset{w}{\Rightarrow} (X,Y)$ and $W_n \overset{w}{\Rightarrow} W$. 

Observe that $I(X;Y)=D(p_{X,Y}||q_{X,Y})$, where $q_{X,Y}=p_Xp_Y$. For the theorem to hold we need to check the assumptions.
First, from Theorem~\ref{thm:weakconvergence}, we have $p_{X_n,Y_n} \overset{w}{\Rightarrow} p_{X,Y}.$
Second, since the marginal distributions converge weakly if the joint distribution converges weakly, we also have $q_{X_n,Y_n} \overset{w}{\Rightarrow} q_{X,Y}.$
Therefore,
\begin{align} \label{eqn:semicont1}
I(X;Y) \leq \liminf\limits_{n \to \infty} I(X_n;Y_n).
\end{align}
To preserve the covariance matrix $K_{(X,Y)}$, there are three degrees of freedom plus one degree of freedom coming from minimizing the objective, thus $|\mathcal{W}| \leq 4$ is enough to attain the minimum.  

Let us introduce $\delta > 0$ and define $N_{\delta}\sim \mathcal{N}(0,\delta)$, being independent of $\{X_n\},X,\{Y_n\}$ and $Y$. From the entropy power inequality, we have 
\begin{align} \label{eqn:semicont2}
h(X_n+N_{\delta}) &\geq h(X_n) \\
h(Y_n+N_{\delta}) &\geq h(Y_n),
\end{align}
and moreover, for Gaussian perturbations, we have
\begin{align} \label{eqn:semicont3}
\liminf_{n\to \infty} h(X_n+N_{\delta}) =h(X+N_{\delta}).
\end{align}
This results in 
\begin{align}
\liminf_{n \to \infty} \ell_{\lambda}(X_n,Y_n) &=\liminf_{n \to \infty} (1+\lambda)I(X_n;Y_n)-\lambda[h(X_n)+h(Y_n)] \\
& \geq \liminf_{n \to \infty} (1+\lambda)I(X_n;Y_n)-\lambda[h(X_n+N_{\delta})+h(Y_n+N_{\delta})] \label{eqn:EPIbound} \\
& \geq (1+\lambda)I(X;Y)-\lambda[h(X+N_{\delta})+h(Y+N_{\delta})], \label{eqn:semicont4}
\end{align}
where (\ref{eqn:EPIbound}) follows from (\ref{eqn:semicont2}) and (\ref{eqn:semicont4}) follows from (\ref{eqn:semicont1}), (\ref{eqn:semicont3}). Letting $\delta \to 0$, we obtain the weak semicontinuity of our object $\liminf_{n \to \infty} \ell_{\lambda}(X_n,Y_n) \geq  \ell_{\lambda}(X,Y)$.

\subsection{A Gaussian auxiliary $W$ attains the infimum in Equation~\eqref{eqn:convexenv_object}}\label{app:HyperGauss:Gauss}
This proof closely follows the arguments in~\cite{Hyper_Gauss}.
We include it for completeness.
We start by creating two identical and independent copies of the minimizer $(W,X,Y)$, which are $(W_1,X_1,Y_1)$ and $(W_2,X_2,Y_2)$. In addition, let us define $X_A=\frac{X_1+X_2}{\sqrt{2}}$, $X_B=\frac{X_1-X_2}{\sqrt{2}}$, $Y_A=\frac{Y_1+Y_2}{\sqrt{2}}$ and $Y_B=\frac{Y_1-Y_2}{\sqrt{2}}$. Thus, we have
\begin{align}
2V(K)&= h(X_1, X_2|W_1, W_2)+ h(Y_1, Y_2|W_1, W_2) -(1+\lambda)h(X_1, X_2, Y_1, Y_2|W_1, W_2) \\
&= h(X_A, X_B|W_1, W_2)+ h(Y_A, Y_B|W_1, W_2) -(1+\lambda)h(X_A, X_B, Y_A, Y_B|W_1, W_2) \label{eqn:entpreserve} \\
&= h(X_A|W_1, W_2)+ h(Y_A|W_1, W_2) -(1+\lambda)h(X_A,Y_A|W_1, W_2) \\
& \quad \quad + h(X_B|X_A,Y_A,W_1, W_2)+ h(Y_B|X_A,Y_A,W_1, W_2) -(1+\lambda)h(X_B, Y_B|X_A,Y_A,W_1, W_2) \\
& \quad \quad + I(X_A; Y_B|Y_A, W_1, W_2) +I(Y_A; X_B|X_A, W_1, W_2) \\
& \geq 2V(K) + I(X_A; Y_B|Y_A, W_1, W_2) +I(Y_A; X_B|X_A, W_1, W_2), \label{eqn:optimizerbound}
\end{align}
where (\ref{eqn:entpreserve}) follows from entropy preservation under bijective transformation and (\ref{eqn:optimizerbound}) follows from definition of $V(K)$ such that $K_{(X_A,Y_A)} \preceq K$. This would imply that 
\begin{align} \label{eqn:alternatezero1}
I(X_A; Y_B|Y_A, W_1, W_2)=I(Y_A; X_B|X_A, W_1, W_2)&=0.
\end{align}
Similarly we get 
\begin{align} \label{eqn:alternatezero2}
I(X_B; Y_A|Y_B, W_1, W_2)=I(Y_B; X_A|X_B, W_1, W_2)&=0,
\end{align}
by switching the roles of index $_A$ and $_B$. Through another way of factorization we get
\begin{align}
2V(K)&= h(X_1, X_2|W_1, W_2)+ h(Y_1, Y_2|W_1, W_2) -(1+\lambda)h(X_1, X_2, Y_1, Y_2|W_1, W_2) \\
&= h(X_A, X_B|W_1, W_2)+ h(Y_A, Y_B|W_1, W_2) -(1+\lambda)h(X_A, X_B, Y_A, Y_B|W_1, W_2) \\
&= h(X_A|W_1, W_2)+ h(Y_A|W_1, W_2) -(1+\lambda)h(X_A,Y_A|W_1, W_2) \\
& \quad \quad + h(X_B|W_1, W_2)+ h(Y_B|W_1, W_2) -(1+\lambda)h(X_B, Y_B|W_1, W_2) \\
& \quad \quad -I(Y_A; Y_B|W_1, W_2) +(1+\lambda)I(X_A, Y_A; X_B, Y_B|W_1, W_2) -I(X_A; X_B|W_1, W_2) \\
& \geq 2V(K) -I(Y_A; Y_B|W_1, W_2) +(1+\lambda)I(X_A, Y_A; X_B, Y_B|W_1, W_2) -I(X_A; X_B|W_1, W_2), 
\end{align}
which implies that 
\begin{align} \label{eqn:thirdapproachfactortize}
-I(Y_A; Y_B|W_1, W_2) +(1+\lambda)I(X_A, Y_A; X_B, Y_B|W_1, W_2) -I(X_A; X_B|W_1, W_2) \leq 0.
\end{align}

Yet, through another way of factorization we get
\begin{align}
2V(K)&= h(X_1, X_2|W_1, W_2)+ h(Y_1, Y_2|W_1, W_2) -(1+\lambda)h(X_1, X_2, Y_1, Y_2|W_1, W_2) \\
&= h(X_A, X_B|W_1, W_2)+ h(Y_A, Y_B|W_1, W_2) -(1+\lambda)h(X_A, X_B, Y_A, Y_B|W_1, W_2) \\
&= h(X_A|X_B,Y_B,W_1, W_2)+ h(Y_A|X_B,Y_B,W_1, W_2) -(1+\lambda)h(X_A,Y_A|X_B,Y_B,W_1, W_2) \\
& \quad \quad + h(X_B|X_A,Y_A,W_1, W_2)+ h(Y_B|X_A,Y_A,W_1, W_2) -(1+\lambda)h(X_B, Y_B|X_A,Y_A,W_1, W_2) \\
& \quad \quad + I(X_B, Y_B; Y_A|W_1, W_2) + I(X_A; Y_B|Y_A, W_1, W_2) -(1+\lambda)I(X_B, Y_B; X_A, Y_A|W_1, W_2) \\
& \quad \quad + I(Y_A; X_B|X_A, W_1, W_2) + I(X_B, Y_B; X_A|W_1, W_2) \\
& \geq 2V(K) + I(X_B, Y_B; Y_A|W_1, W_2) + I(X_A; Y_B|Y_A, W_1, W_2) \\
& \quad \quad -(1+\lambda)I(X_B, Y_B; X_A, Y_A|W_1, W_2) + I(Y_A; X_B|X_A, W_1, W_2) + I(X_B, Y_B; X_A|W_1, W_2),
\end{align}
which implies that 
\begin{align} 
& I(X_B, Y_B; Y_A|W_1, W_2) + I(X_A; Y_B|Y_A, W_1, W_2) -(1+\lambda)I(X_B, Y_B; X_A, Y_A|W_1, W_2)\\
& \quad \quad  + I(Y_A; X_B|X_A, W_1, W_2) + I(X_B, Y_B; X_A|W_1, W_2) \leq 0.
\end{align}
By substituting (\ref{eqn:alternatezero1}) and (\ref{eqn:alternatezero2}), we obtain
\begin{align} \label{eqn:simplifysecondapproach}
& I(Y_B; Y_A|W_1, W_2)  + I(X_B; X_A|W_1, W_2) -(1+\lambda)I(X_B, Y_B; X_A, Y_A|W_1, W_2) \leq 0.
\end{align}
Equation (\ref{eqn:thirdapproachfactortize}) and (\ref{eqn:simplifysecondapproach}) imply that 
\begin{align} \label{eqn:simplifyequalityzerothreeapproaches}
& I(Y_B; Y_A|W_1, W_2)  + I(X_B; X_A|W_1, W_2) -(1+\lambda)I(X_B, Y_B; X_A, Y_A|W_1, W_2) =0.
\end{align}

Another way of factorizing is the follows 
\begin{align}
2V(K)&= h(X_1, X_2|W_1, W_2)+ h(Y_1, Y_2|W_1, W_2) -(1+\lambda)h(X_1, X_2, Y_1, Y_2|W_1, W_2) \\
&= h(X_A, X_B|W_1, W_2)+ h(Y_A, Y_B|W_1, W_2) -(1+\lambda)h(X_A, X_B, Y_A, Y_B|W_1, W_2) \\
&= h(X_A|X_B,Y_B,W_1, W_2)+ h(Y_A|X_B,Y_B,W_1, W_2) -(1+\lambda)h(X_A,Y_A|X_B,Y_B,W_1, W_2) \\
& \quad \quad + h(X_B|X_A,W_1, W_2)+ h(Y_B|X_A,W_1, W_2) -(1+\lambda)h(X_B, Y_B|X_A,W_1, W_2) \\
& \quad \quad + I(X_A; Y_B|W_1, W_2) + I(Y_A; X_B|Y_B, W_1, W_2) -(1+\lambda)I(X_B, Y_B; X_A|W_1, W_2) \\
& \quad \quad + I(X_B, Y_B; X_A|W_1, W_2) \\
& \geq 2V(K) + I(X_A; Y_B|W_1, W_2) + I(Y_A; X_B|Y_B, W_1, W_2) -(1+\lambda)I(X_B, Y_B; X_A|W_1, W_2) \\
& \quad \quad + I(X_B, Y_B; X_A|W_1, W_2),
\end{align}
which implies that 
\begin{align}
I(X_A; Y_B|W_1, W_2) + I(Y_A; X_B|Y_B, W_1, W_2) -\lambda I(X_B, Y_B; X_A|W_1, W_2) \leq 0.
\end{align}
Using $(\ref{eqn:alternatezero1})$ and $(\ref{eqn:alternatezero2})$, we simplify the above equation into
\begin{align}
I(X_A; Y_B|W_1, W_2) -\lambda I(X_B; X_A|W_1, W_2) \leq 0.
\end{align}
Switching the roles of $_A$ and $_B$ we obtain
\begin{align} \label{eqn:mutualinformationcontraction}
I(X_B; Y_A|W_1, W_2) -\lambda I(X_A; X_B|W_1, W_2) \leq 0.
\end{align}

By factorizing for the last time we get 
\begin{align}
2V(K)&= h(X_1, X_2|W_1, W_2)+ h(Y_1, Y_2|W_1, W_2) -(1+\lambda)h(X_1, X_2, Y_1, Y_2|W_1, W_2) \\
&= h(X_A, X_B|W_1, W_2)+ h(Y_A, Y_B|W_1, W_2) -(1+\lambda)h(X_A, X_B, Y_A, Y_B|W_1, W_2) \\
&= h(X_A|X_B,Y_B,W_1, W_2)+ h(Y_A|X_B,Y_B,W_1, W_2) -(1+\lambda)h(X_A,Y_A|X_B,Y_B,W_1, W_2) \\
& \quad \quad + h(X_B|Y_A,W_1, W_2)+ h(Y_B|XY_A,W_1, W_2) -(1+\lambda)h(X_B, Y_B|Y_A,W_1, W_2) \\
& \quad \quad +I(Y_A; X_B, Y_B|W_1, W_2) -(1+\lambda) I(X_B, Y_B; Y_A|W_1, W_2) \\
& \quad \quad + I(Y_B; X_A|W_1, W_2) + I(Y_B; X_A|X_B, W_1, W_2) \\
& \geq 2V(K) -\lambda I(X_B, Y_B; Y_A|W_1, W_2) \\
& \quad  \quad+ I(Y_B; X_A|W_1, W_2) + I(Y_B; X_A|X_B, W_1, W_2),
\end{align}
which implies that 
\begin{align}
-\lambda I(X_B, Y_B; Y_A|W_1, W_2) + I(Y_B; X_A|W_1, W_2) + I(Y_B; X_A|X_B, W_1, W_2) \leq 0.
\end{align}
The above inequality is simplified by using (\ref{eqn:alternatezero1}) and (\ref{eqn:alternatezero2}) into 
\begin{align}
-\lambda I(Y_B; Y_A|W_1, W_2) + I(Y_B; X_A|W_1, W_2) \leq 0.
\end{align}
Switching the roles of $_A$ and $_B$ we obtain
\begin{align} \label{eqn:mutualinfocontraction2}
-\lambda I(Y_A; Y_B|W_1, W_2) + I(Y_A; X_B|W_1, W_2) \leq 0.
\end{align}
Observe that
\begin{align}
I(X_A; Y_B|X_B, Y_A, W_1, W_2) &= I(X_A; Y_B, X_B|Y_A, W_1, W_2)- I(X_A; X_B|Y_A, W_1, W_2) \\
&=I(X_A, Y_A; Y_B, X_B|W_1, W_2) -I(Y_A; Y_B, X_B|W_1, W_2) \\
& \quad \quad -I(X_A, Y_A; X_B|W_1, W_2) + I(Y_A; X_B|W_1, W_2) \\
&=I(X_A, Y_A; Y_B, X_B|W_1, W_2) -I(Y_A; Y_B|W_1, W_2) \\
& \quad \quad -I(X_A; X_B|W_1, W_2) + I(Y_A; X_B|W_1, W_2)  -I(Y_A; Y_B|W_1, W_2)   \label{eqn:applicationalternate}\\
&= \left( \frac{1}{1+\lambda}-1 \right) I(X_A; X_B|W_1, W_2) + I(Y_A; X_B|W_1, W_2) \\
& \quad \quad + \left( \frac{1}{1+\lambda}-1 \right) I(Y_A; Y_B|W_1, W_2)  \geq 0 \label{eqn:abcd} 
\end{align}
where~\eqref{eqn:applicationalternate} follows from~\eqref{eqn:alternatezero1} and~\eqref{eqn:alternatezero2}, Equation~\eqref{eqn:abcd} follows from (\ref{eqn:simplifyequalityzerothreeapproaches}) and from non-negativity of the conditional mutual information term we obtain 
\begin{align} \label{eqn:cdef}
\left( \frac{1}{1+\lambda}-1 \right) I(X_A; X_B|W_1, W_2) + I(Y_A; X_B|W_1, W_2) + \left( \frac{1}{1+\lambda}-1 \right) I(Y_A; Y_B|W_1, W_2)  \geq 0.
\end{align}
Using $(\ref{eqn:mutualinformationcontraction})$, the above equation simplifies into 
\begin{align}
\left( \frac{1}{1+\lambda}-1 \right) I(X_A; X_B|W_1, W_2) + \lambda I(X_A; X_B|W_1, W_2) + \left( \frac{1}{1+\lambda}-1 \right) I(Y_A; Y_B|W_1, W_2)  \geq 0,
\end{align}
which simplifies into
\begin{align} \label{eqn:XtoYcontraction}
\lambda I(X_A; X_B|W_1, W_2) \geq I(Y_A; Y_B|W_1, W_2),
\end{align}
form the fact that $\lambda > 0$. Using Equation~\eqref{eqn:mutualinfocontraction2} in~\eqref{eqn:cdef} we obtain
\begin{align} 
\left( \frac{1}{1+\lambda}-1 \right) I(X_A; X_B|W_1, W_2) + \lambda I(Y_A; Y_B|W_1, W_2) + \left( \frac{1}{1+\lambda}-1 \right) I(Y_A; Y_B|W_1, W_2)  \geq 0,
\end{align}
which simplifies into
\begin{align} \label{eqn:YtoXcontraction}
\lambda I(Y_A; Y_B|W_1, W_2) \geq I(X_A; X_B|W_1, W_2).
\end{align}
Combining (\ref{eqn:XtoYcontraction}) and (\ref{eqn:YtoXcontraction}) we obtain 
\begin{align} 
(1-\lambda^2) I(Y_A; Y_B|W_1, W_2) &\leq 0, \\
(1-\lambda^2) I(X_A; X_B|W_1, W_2) &\leq 0.
\end{align}
Knowing that $0 < \lambda < 1$, we deduce that $I(Y_A; Y_B|W_1, W_2)=I(X_A; X_B|W_1, W_2)=0$. Finally using (\ref{eqn:simplifyequalityzerothreeapproaches}) we obtain
\begin{align} \label{eqn:maindoublingtrick}
I(X_B, Y_B; X_A, Y_A|W_1, W_2) =0.
\end{align}
In addition, let us introduce the following theorem before arriving to the concluding result.
\begin{theorem}[Corollary to Theorem 1 in \cite{Ghurye--Olkin1962}] \label{thm:DT}
If $\ve{Z}_1$ and $\ve{Z}_2$ are independent multidimensional random column vectors, and if $\ve{Z}_1+\ve{Z}_2$ and $\ve{Z}_1-\ve{Z}_2$ are independent then $\ve{Z}_1$, $\ve{Z}_2$ are normally distributed with identical covariances.
\end{theorem}
Thus, the following statements are true
\begin{itemize}
\item The pair $(X_1, Y_1)$ and $(X_2, Y_2)$ are conditionally independent given $W_1 = w_1, W_2 = w_2$ from assumption.
\item The pair $(\frac{X_1+X_2}{\sqrt{2}}, \frac{Y_1+Y_2}{\sqrt{2}})$ and $(\frac{X_1-X_2}{\sqrt{2}}, \frac{Y_1-Y_2}{\sqrt{2}})$ are conditionally independent given $W_1 = w_1, W_2 = w_2$. This follows from (\ref{eqn:maindoublingtrick}).
\end{itemize}

By applying Theorem~\ref{thm:DT} to the fact regarding conditional independence established in Equation~\eqref{eqn:maindoublingtrick}, we can infer that $(X,Y)| \{ W=w\} \sim \mathcal{N}(0,K_w)$, where $K_w$ might depend on the realization of $W=w$. We will now argue that this is not the case. To make a brief summary we have shown the existence part thus, by choosing $W$ to be the trivial random variable a single Gaussian (i.e. not a Gaussian mixture) is one of the possible minimizers. Let us suppose that there are two Gaussian minimizers $\mathcal{N}(0,K_{w_1})$ and $\mathcal{N}(0,K_{w_2})$, where $K_{w_1} \neq K_{w_2}$. Consider the random variable $(W,X,Y)$ where, $(X,Y)|\{ W=w_1\} \sim \mathcal{N}(0,K_{w_1})$ and $(X,Y)|\{ W=w_2\} \sim \mathcal{N}(0,K_{w_2})$. Therefore the triple $(W,X,Y)$ also attains $V(K)$ and satisfies the covariance constraint. At the same time, we showed that the sum and the difference are also minimizers and they must be independent of each other, which happens only when $K_{w_1}= K_{w_2}$. In other words, $K_w$ does not depend on the realization $W=w$, and the $(W,X,Y)$ is a single Gaussian (i.e. not a Gaussian mixture). We established that $(W,X,Y)$ is a unique Gaussian minimizer, and thus
\begin{align}
V(K)=h(X|W)+h(Y|W)-(1+\lambda)h(X,Y|W).
\end{align}
Furthermore, there exists a decomposition $(X,Y)=W+(X^{\prime},Y^{\prime}) \sim \mathcal{N}(0,K)$, where $W$ is independent of $(X^{\prime},Y^{\prime})$ and $W \sim \mathcal{N}(0,K-K^{\prime})$ and $(X^{\prime},Y^{\prime}) \sim \mathcal{N}(0,K^{\prime})$. Then,
\begin{align}
V(K)=h(X^{\prime})+h(Y^{\prime})-(1+\lambda)h(X^{\prime},Y^{\prime}),
\end{align}
thus establishing Theorem~\ref{Thm:Hypercontract}, because
\begin{align}
V(K) \leq \inf_W h(X|W)+h(Y|W)-(1+\lambda)h(X,Y|W),
\end{align}
by the definition of $V(K)$.

\section{Proof of Lemma \ref{lem:lemmasymmery}} \label{App:lowerboundmainRWCI}

Let $K^{\prime}$ be parametrized as $K^{\prime} = \begin{pmatrix} \sigma^2_X & q\sigma_X \sigma_Y \\ q\sigma_X \sigma_Y & \sigma^2_Y  \end{pmatrix} \succeq 0$. Then, the problem is the same to the following one

\begin{align}
&\min_{K^{\prime}: 0 \preceq K^{\prime} \preceq \begin{pmatrix} 1 & \rho \\ \rho &1  \end{pmatrix}} h(X^{\prime})+h(Y^{\prime})-(1+\lambda)h(X^{\prime},Y^{\prime}) \nonumber \\
& \quad \quad \quad = \min_{(\sigma_X,\sigma_Y,q) \in \mathcal{A}_{\rho}} \frac{1}{2}\log{(2 \pi e)^2 \sigma_X^2\sigma_Y^2} -\frac{1+\lambda}{2}\log{(2 \pi e)^2 \sigma_X^2\sigma_Y^2(1-q^2)} \label{eqn:mlproof1}
\end{align}
where the set
\begin{align} 
\mathcal{A}_{\rho}=\left\{( \sigma_X,\sigma_Y,q): \begin{pmatrix} \sigma^2_X -1& q\sigma_X \sigma_Y -\rho\\ q\sigma_X \sigma_Y-\rho & \sigma^2_Y -1 \end{pmatrix} \preceq 0 \right\}.
\end{align}
Matrices of dimension $2\times2$ are negative semi-definite if and only if the trace is negative and determinant is positive. Thus, we can rewrite the set as 
\begin{align} 
\mathcal{A}_{\rho}=\left\{(\sigma_X,\sigma_Y,q): \substack{\sigma^2_X+\sigma^2_Y \leq 2, \\ \quad (1-q^2)\sigma^2_X\sigma^2_Y +2\rho q \sigma_X\sigma_Y +1-\rho^2-(\sigma^2_X+\sigma^2_Y) \geq 0} \right\}.
\end{align}
By making use of $\sigma^2_X+\sigma^2_Y \geq 2\sigma_X\sigma_Y$, we derive that $\mathcal{A}_{\rho} \subseteq \mathcal{B}_{\rho}$, where
\begin{align} 
\mathcal{B}_{\rho}=\left\{ (\sigma_X,\sigma_Y,q): \substack{\sigma_X\sigma_Y \leq 1, \\ \quad (1-q^2)\sigma^2_X\sigma^2_Y +2\rho q \sigma_X\sigma_Y +1-\rho^2-2\sigma_X\sigma_Y) \geq 0} \right\}.
\end{align} We will further reparametrize and define $\sigma^2=\sigma_X\sigma_Y$, thus 
\begin{align} 
\mathcal{D}_{\rho}=\left\{( \sigma^2,q): \substack{\sigma^2 \leq 1, \\ (\sigma^2(1-q)-1+\rho)(\sigma^2(1+q)-1-\rho) \geq 0} \right\}.
\end{align}
The second equation in the definition of the set $\mathcal{D}_{\rho}$ has roots $\sigma^2=\frac{1+\rho}{1+q}$ and $\sigma^2=\frac{1-\rho}{1-q}$, thus the inequality is true if $\sigma^2$ is not in between these two roots. Thus, we can rewrite the set $\mathcal{D}_{\rho}$ as
\begin{align} 
\mathcal{D}_{\rho}=\left\{( \sigma^2,q): \substack{\rho \geq q, \quad \sigma^2(1-q) \leq 1-\rho \\ \rho < q, \quad  \sigma^2(1+q) \leq 1+\rho } \label{eqn:D} \right\}.
\end{align}
Thus, we have
\begin{align}
\min_{(\sigma_X,\sigma_Y,q) \in \mathcal{A}_{\rho}} \frac{1}{2}\log{(2 \pi e)^2 \sigma_X^2\sigma_Y^2} -\frac{1+\lambda}{2}\log{(2 \pi e)^2 \sigma_X^2\sigma_Y^2(1-q^2)} \geq \min_{(\sigma^2,q) \in \mathcal{D}_{\rho}} f(\lambda,\sigma^2,q)
\label{eqn:mlproof2}
\end{align}
where,
\begin{align}f(\lambda,\sigma^2,q)&=\frac{1}{2}\log{(2 \pi e)^2 \sigma^4}-\frac{1+\lambda}{2}\log{(2 \pi e)^2 \sigma^4(1-q^2)} \label{eqn:f}.
\end{align}
For now let us assume $\rho$ is positive and start from the case $\rho \geq q$. Then, by weak duality we have %The function $f$ is decreasing in $\sigma^2$. Also, the function $f$ is convex in $\sigma^2$. Since the object is continuous in $\sigma^2$ and the constraint is linear for any fixed $q$, then the optimal choice is $\sigma^2=\frac{1-\rho}{1-q}$. Thus, 
\begin{align}
\min\limits_{(\sigma^2,q) \in \mathcal{D}_{\rho}} f(\lambda,\sigma^2,q) \geq \min\limits_{\sigma^2,q} f(\lambda,\sigma^2,q) + \mu(\sigma^2(1-q)-1+\rho)), \label{eqn:mlproof3}
\end{align}
for any $\mu \geq 0$.
By applying Karush-Kuhn-Tucker (KKT) conditions we get
\begin{align}
\frac{\partial }{\partial \sigma^2}=-\frac{\lambda}{\sigma^2} + \mu (1-q)&=0, \label{eqn:KKT1} \\ 
\frac{\partial }{\partial q}=\frac{(1+\lambda)q}{1-q^2} - \mu \sigma^2&=0, \label{eqn:KKT2} \\
\mu(\sigma^2(1-q)-1+\rho))&=0, \label{eqn:KKT3}
\end{align}
where (\ref{eqn:KKT1}), (\ref{eqn:KKT2}) is known as stationary condition and (\ref{eqn:KKT3}) is known as complementary slackness condition. By using (\ref{eqn:KKT1}) we get 
\begin{align}
\mu=\frac{\lambda}{\sigma^2(1-q)} \label{eqn:KKT4}.
\end{align}
By using (\ref{eqn:KKT2}) we get 
\begin{align}
\mu=\frac{(1+\lambda)q}{\sigma^2(1-q^2)} \label{eqn:KKT5}.
\end{align}
By equating (\ref{eqn:KKT4}) and (\ref{eqn:KKT5}) we deduce that $q_*=\lambda$. Since $\lambda>0$, then $\mu \neq 0$ and by using (\ref{eqn:KKT3}) we get $\sigma^2_*=\frac{1-\rho}{1-\lambda}$. In addition, $\mu_*=\frac{\lambda}{1-\rho}$. Since the KKT conditions are satisfied by $q_*,\sigma^2_*$ and $\mu_*$ then strong duality holds, thus 
\begin{align}
\min\limits_{\sigma^2,q} f(\lambda,\sigma^2,q) + \mu(\sigma^2(1-q)-1+\rho))&= f(\lambda,\frac{1-\rho}{1-\lambda},\lambda) \nonumber \\
&= \frac{1}{2} \log{\frac{1}{1-\lambda^2}}-\frac{\lambda}{2} \log{(2\pi e)^2\frac{(1-\rho)^2(1+\lambda)}{1-\lambda}}. \label{eqn:mlprooffinal}
\end{align}
By combining (\ref{eqn:mlproof1}), (\ref{eqn:mlproof2}), (\ref{eqn:mlproof3}) and (\ref{eqn:mlprooffinal}) we get the desired lower bound. 

For the case $\rho < q$, let us optimize over $\sigma^2$ for any fixed $q$. The function $f$ is decreasing in $\sigma^2$. Also, the function $f$ is convex in $\sigma^2$. Since the object is continuous in $\sigma^2$ and the constraint is linear for any fixed $q$, then the optimal choice is $\sigma^2=\frac{1+\rho}{1+q}$. Thus, 
\begin{align}
\min\limits_{(\sigma^2,q) \in \mathcal{D}_{\rho}} f(\lambda,\sigma^2,q) \geq \min\limits_{q \in [\rho, 1]} f(\lambda,\frac{1+\rho}{1+q},q). \label{eqn:mlproof4}
\end{align}
The function on the right hand side can be written as
\begin{align}
f(\lambda,\frac{1+\rho}{1+q},q)&= \frac{1}{2}\log{(2\pi e)^2\frac{(1+\rho)^2}{(1+q)^2}} -\frac{1+\lambda}{2}\log{(2\pi e)^2\frac{(1+\rho)^2(1-q)}{(1+q)}}.
\end{align}
The function is convex and increasing in $q$ for $q \in [\rho, 1]$, 
\begin{align}
\frac{\partial f}{\partial q}&=\frac{q+\lambda}{1-q^2}>0, \\
\frac{\partial^2 f}{\partial q^2}&=\frac{1+q^2+2\lambda q}{(1-q^2)^2} >0,
\end{align}
thus, the optimal value of $q=\rho$, is guaranteed to give the minimum. To conclude we show that $f(\lambda,\frac{1-\rho}{1-\lambda},\lambda) \leq f(\lambda,1,\rho)$ for $\lambda \leq \rho$. To show this we define
\begin{align}
h(\lambda)&:=f(\lambda,\frac{1-\rho}{1-\lambda},\lambda)-f(\lambda,1,\rho)\\
&=\frac{1}{2}\log{\frac{1-\rho^2}{1-\lambda^2}}- \frac{\lambda}{2}\log{\frac{(1+\lambda)(1-\rho)}{(1-\lambda)(1+\rho)}},
\end{align}
and the new defined function is increasing in $\lambda$,
\begin{align}
\frac{\partial h}{\partial \lambda}&=-\frac{1}{2} \log{\frac{(1+\lambda)(1-\rho)}{(1-\lambda)(1+\rho)}} \geq 0, \quad \text{for } \lambda\leq \rho 
\end{align}
and it is concave in $\lambda$,
\begin{align}
\frac{\partial^2 h}{\partial \lambda^2}&=-\frac{1}{1-\lambda^2}<0,
\end{align}
thus, $h(\lambda) \leq h(\rho)=0$. Then, $f(\lambda,1,\rho) \geq f(\lambda,\frac{1-\rho}{1-\lambda},\lambda)$. The argument goes through also for the case when $\rho$ is negative, which completes the proof.

%\section{Proof of Lemma \ref{lemma:GrayWynerComb}} \label{App:GrayWynerNetwork}
%We start by bounding as follows
%\begin{align} 
%&h(X,Y) - \nu \alpha -\nu \log{(2 \pi e D)}  + \nu \cdot \min_{0 \preceq K^{\prime} \preceq \begin{pmatrix} 1 & \rho \\ \rho &1  \end{pmatrix}} h(X^{\prime})+h(Y^{\prime})- \frac{1}{\nu}h(X^{\prime},Y^{\prime})  \\
%& \geq h(X,Y) - \nu \alpha -\nu \log{(2 \pi e D)}  + \nu \cdot \min_{(\sigma^2,q) \in \mathcal{D}_{\rho}} f(\frac{1}{\nu}-1,\sigma^2,q)
%\end{align}
%where,
%\begin{align}
%\mathcal{D}_{\rho}&=\left\{ (\sigma^2,q): \substack{\rho>q, \quad \sigma^2(1-q) \leq 1-\rho \\ \rho \leq q, \quad  \sigma^2(1+q) \leq 1+\rho } \right\}.
%\end{align}
%and function $f$ is defined in (\ref{eqn:f}) and $\mathcal{D}_{\rho}$ is defined in (\ref{eqn:setD}).

\section*{Acknowledgment}
This work was supported in part by the Swiss National Science Foundation under Grant 169294, Grant P2ELP2\_165137. 

\bibliographystyle{IEEEtran}
\bibliography{nit_wyn,caching_giel}


\end{document}